\documentclass[11pt]{article}

\usepackage{arxiv}
\usepackage{float}
\usepackage{amsmath,amssymb}
\usepackage[utf8]{inputenc} 
\usepackage[T1]{fontenc}    
\usepackage{hyperref}       
\usepackage{url}            
\usepackage{booktabs}       
\usepackage{amsfonts,amsthm}       
\usepackage{nicefrac}       
\usepackage{microtype}      
\usepackage{lipsum}
\usepackage{graphicx}
\graphicspath{ {./images/} }
\usepackage{bbm}

\def\Pr{\textsf{Pr}}

\providecommand{\ignore}[1]{}
\newcommand{\tail}{\mathsf{tail}}
\newcommand{\maxsgn}{\mathsf{Max\pm}}

\newcommand{\calA}{\mathcal{A}}
\newcommand{\calD}{\mathcal{D}}
\newcommand{\calE}{\mathcal{E}}
\newcommand{\eps}{\varepsilon}

\newcommand{\sgn}{\mathrm{sgn}}
\newcommand{\Ex}{\mathop{\mathbb{E}}}
\newcommand{\Otilde}{\widetilde{O}}
\newcommand{\Omegatilde}{\widetilde{\Omega}}

\newcommand{\Lap}{\mathrm{Lap}}
\newcommand{\calB}{\mathcal{B}}

\renewcommand{\v}{\boldsymbol{v}}
\newcommand{\bro}{\boldsymbol{\rho}}
\newcommand{\collected}{{\rm Collected}}

\usepackage{bbm}  
\newcommand{\1}{\mathbbm{1}}

\usepackage[algo2e, ruled, vlined]{algorithm2e}
\usepackage{xcolor} 

\SetCommentSty{mycommfont}

\usepackage{dsfont}
\SetArgSty{textnormal}
\usepackage{soul,xcolor}
\newcommand{\remove}[1]{}

\newcommand{\msnote}[1]{{\bf{\color{blue}[Moshe: #1]}}}
\newcommand{\ecnote}[1]{{\bf{\color{purple}[Edith: #1]}}}

\DeclareMathSymbol{\N}{\mathbin}{AMSb}{"4E}
\DeclareMathSymbol{\Z}{\mathbin}{AMSb}{"5A}
\DeclareMathSymbol{\R}{\mathbin}{AMSb}{"52}
\DeclareMathSymbol{\Q}{\mathbin}{AMSb}{"51}
\DeclareMathSymbol{\erert}{\mathbin}{AMSb}{"50}
\DeclareMathSymbol{\I}{\mathbin}{AMSb}{"49}
\DeclareMathSymbol{\C}{\mathbin}{AMSb}{"43}
\DeclareMathSymbol{\E}{\mathbin}{AMSb}{"45}

\def\countsketch{\text{\texttt{CountSketch}}}
\def\sketch{\text{\texttt{Sketch}}}
\def\Bcountsketch{\text{\texttt{BCountSketch}}}
\def\ThresholdMonitor{\text{\texttt{ThresholdMonitor}}}
\def\FastQueryVariant{\text{\texttt{FastQueryVariant}}}
\def\FastQueryHH{\text{\texttt{FastQueryHH}}}

\def\argmax{\mbox{\rm argmax}}

\newcommand{\sign}		{{\rm sign}}

\newcommand{\Bin}		{{\rm Bin}}

\newcommand{\Var}		{{\rm Var}}

\newcommand{\Supp}		{{\rm Supp}}
\newcommand{\RW}		{{\rm RW}}
\newcommand{\BLRW}		{{\rm BLRW}}

\newtheorem{thm}{Theorem}[section]
\newtheorem{theorem}{Theorem}[section]

\newtheorem{lemma}[thm]{Lemma}

\newtheorem{claim}[thm]{Claim}
\newtheorem{remark}[thm]{Remark}

\newtheorem{definition}[thm]{Definition}

\newtheorem{corollary}[thm]{ Corollary}

\newtheorem{question}[thm]{Question}

\usepackage{enumitem}
\begin{document}

\title{On the Robustness of CountSketch to Adaptive Inputs} 

\author{{\normalfont Edith Cohen}\thanks{Google Research and Tel Aviv University. \texttt{edith@cohenwang.com}.} \and Xin Lyu\thanks{UC Berkeley. \texttt{lyuxin1999@gmail.com}.} \and Jelani Nelson\thanks{UC Berkeley and Google Research. \texttt{minilek@alum.mit.edu}.} \and Tam\'{a}s Sarl\'{o}s\thanks{Google Research. \texttt{stamas@google.com}.} \and Moshe Shechner\thanks{Tel Aviv University. \texttt{moshe.shechner@gmail.com}. Partially supported by the Israel Science Foundation (grant 1871/19).} \and Uri Stemmer\thanks{Tel Aviv University and Google Research. \texttt{u@uri.co.il}. Partially supported by the Israel Science Foundation (grant 1871/19)
and by Len Blavatnik and the Blavatnik Family foundation.}}

\maketitle

 \begin{abstract}
\texttt{CountSketch} is a popular dimensionality reduction technique that maps vectors to a lower dimension using
randomized linear measurements.  The sketch supports recovering $\ell_2$-heavy hitters of a vector (entries with $v[i]^2 \geq \frac{1}{k}\|\boldsymbol{v}\|^2_2$). We study the robustness of the sketch in {\em adaptive} settings where input vectors may depend on the output from prior inputs.  Adaptive settings arise in processes with feedback or with adversarial attacks.  We show that the classic estimator is not robust, and can be attacked with a number of queries of the order of the sketch size.  We propose a robust estimator (for a slightly modified sketch) that allows for quadratic number of queries in the sketch size, which is an improvement factor of $\sqrt{k}$ (for $k$ heavy hitters) over prior work. 
 \end{abstract}

\section{Introduction}
\ignore{Prior version:
Sketching and streaming algorithms are often analyzed under the assumption that their internal randomness is independent of their inputs. This assumption, however, is not always reasonable. For example, consider a large system in which a sketching algorithm is used to analyze data coming from one part of the system while answering queries generated by another part of the system, but these (supposedly) different parts of the system are connected via a feedback loop.  In such a case, it is no longer true that the inputs are generated independently of algorithm's randomness (as the inputs depend on previous outputs which depend on the internal randomness) and hence classical algorithms might fail to provide meaningful utility guarantees. 

This motivated a growing interest in designing {\em robust algorithms} which maintain utility even when their inputs are chosen adaptively, possibly as a function of their previous outputs. Works in this vein include \cite{BenEliezerJWY21,HassidimKMMS20,WoodruffZ21,AttiasCSS21,BEO21}. However, the resulting (robust) algorithms are generally significantly less efficient then their classical counterparts. Furthermore, while the analyses of classical algorithms do not seem to carry over to the robust setting, so far we do not have any example showing that they are not robust (i.e,. for all we know it might only be the analysis that breaks).\footnote{It is worth mentioning that \cite{BenEliezerJWY21} showed an attack on a simplified version of the classical AMS sketch. However, their attack does not apply to the AMS sketch itself.}  This naturally raises the following question.

\begin{question}
Could it be that classical sketching and streaming algorithms are in fact robust (as is)?
\end{question}

We believe that understanding this question is crucial for justifying the work on robust algorithms.}

Algorithms are often analyzed and used under the assumption that their internal randomness is independent of their inputs. This assumption, however, is not always reasonable. For example, consider a large system where there is a feedback loop between inputs and outputs or an explicit adversarial attack aimed at constructing inputs on which the system fails.  In such a case, it is no longer true that the inputs are generated independently of the algorithm's randomness (as the inputs depend on previous outputs which depend on the internal randomness) and hence the algorithms might fail to provide utility.

This motivated a growing interest in understanding the performance of algorithms when the inputs are chosen {\em adaptively}, possibly as a function of their previous outputs, and in designing {\em robust algorithms} which provide utility guarantees even when their inputs are chosen adaptively. Works in this vein
span multiple areas, including machine learning~\cite{szegedy2013intriguing,goodfellow2014explaining,athalye2018synthesizing,papernot2017practical}, adaptive data analysis~\cite{Freedman:1983,Ioannidis:2005,FreedmanParadox:2009,DworkFHPRR15}, dynamic graph algorithms~\cite{ShiloachEven:JACM1981,gawrychowskiMW:ICALP2020,GutenbergPW:SODA2020,Wajc:STOC2020}, and sketching and streaming algorithms~\cite{BenEliezerJWY21,HassidimKMMS20,WoodruffZ21,AttiasCSS21,BEO21}.
However, the resulting (robust) algorithms tend to be significantly less efficient then their classical counterparts. Furthermore, while the analyses of classical algorithms do not seem to carry over to the robust setting, in many cases we do not have any example showing that they are not robust (i.e,. for all we know it might only be the analysis that breaks).\footnote{It is worth mentioning that \cite{BenEliezerJWY21} showed an attack on a simplified version of the classical AMS sketch (with a weaker estimator). However, their attack does not apply with the classic estimator.
\cite{HardtW:STOC2013} constructed an attack on linear sketches but the size of the attack is far from respective upper bounds.}
This naturally raises the question of 
quantifying more precisely the robustness of algorithms.

Driven by this question, in this work we set out to explore the robustness properties of the classical \texttt{CountSketch} algorithm~\cite{CharikarCFC:2002}, related to {\em feature hashing}~\cite{MoodyD:NIPS1989} in the machine learning literature.
\texttt{CountSketch} is a popular dimensionality reduction technique that maps vectors to a lower-dimension using randomized linear measurements. The method is textbook material and has found many applications in machine learning and data analysis~\cite{WeinbergerDLSA:ICML2009,pmlr-v5-shi09a,ChenWTWC:ICML2015,KDD-2016-ChenWTWC,AmiraliSLDSB:ICML2018,SpringKMA:ICML2019,AhleKKPVWZ:SODA2020,CPW:NEURIPS2020}. 
The sketch is often a component of large ML models or data analytics systems and its robustness may impact the overall robustness of the system.

Operationally, \texttt{CountSketch} is parametrized by $(n,d,b)$, where $n$ is the dimension of input vectors, $d$ is the size of the sketch, and $b$ is a parameter controlling the accuracy of the sketch (referred to as the ``width'' of the sketch). It is applied by initializing $d/b$ pairs of random hash functions $\boldsymbol{\rho}=\left((h_1,s_1),\dots,(h_{d/b},s_{d/b})\right)$ where  $h_j:[n]\rightarrow[b]$ and $s_j:[n]\rightarrow\{\pm1\}$. We think of $\boldsymbol{\rho}$ as defining $\frac{d}{b}\times b$ ``buckets'' ($b$ buckets for every pair of hash functions). To sketch a vector $\boldsymbol{v}\in\mathbb{R}^n$: For every $i\in[n]$ and for every $j\in[d/b]$, add $s_j(i)\cdot\boldsymbol{v}[i]$ to the bucket indexed by $\left( j , h_j(i)\right)$. The resulting collection of $d$ summations (the values of the buckets) is the sketch, which we denote as $\sketch_{\boldsymbol{\rho}}(\boldsymbol{v})$. That is, $\sketch_{\boldsymbol{\rho}}(\boldsymbol{v}):=
\left(c_{(j,w)}\right)_{j\in[d/b],w\in[b]}$, where
$$
c_{j,w}:=\sum_{i\in[n]: h_j(i)=w} s_j(i)\cdot\boldsymbol{v}[i].
$$

In applications, the desired task 
(e.g., recovering the set of heavy hitter entries of $\boldsymbol{v}$ and approximate values, reconstructing an approximation $\hat{\boldsymbol{v}}$ of the input vector, or approximating the inner product of two vectors) is obtained by applying an {\em estimator} $M$ to the sketch. Note that $M$ does not access the original vector. The most commonly used estimator in the context of \texttt{CountSketch} is the {\em median estimator}, with which the $\ell$th enrty of the original vector is estimated as
$
{\rm Median}_{j\in[d/b]}\left\{
s_j(\ell)\cdot c_{\left(j,h_j(\ell)\right)}
\right\}.
$ 
To intuit this estimator, observe that for every $j\in[d/b]$ we have that
\begin{align*}
s_j(\ell)\cdot c_{\left(j,h_j(\ell)\right)}=
s_j(\ell)\cdot \left(\sum_{i\in[n]:\; h_j(i)=h_j(\ell)} s_j(i)\cdot\boldsymbol{v}[i]\right) 
=
\boldsymbol{v}[\ell]+\left(
\sum_{i\neq\ell:\; h_j(i)=h_j(\ell)} s_j(\ell)\cdot s_j(i)\cdot\boldsymbol{v}[i]\right)
\end{align*}
is an unbiased estimator for $\boldsymbol{v}[\ell]$. Hence, intuitively, the median of these values is a good estimate.

Our focus will be on the task of recovering the $\ell_2$-heavy hitters. An $\ell_2$-heavy hitter with parameter $k$ of a vector $\boldsymbol{v}$ is defined to be an entry $i$ such that $\boldsymbol{v}[i]^2 > \frac{1}{k} \|\boldsymbol{v}_{\tail[k]} \|_2^2$, where $\boldsymbol{v}_{\tail[k]}$, 
the $k$-tail of $\boldsymbol{v}$, is a vector obtained from $\boldsymbol{v}$ by replacing its $k$ largest entries in magnitude with $0$. The {\em heavy-hitters problem} is to return a set of $O(k)$ keys that includes all heavy hitters.
The output is correct when all heavy hitters are reported. In the non-adaptive setting, this problem can be solved using \texttt{CountSketch} with $d=O(k\log n)$ and $b=O(k)$, by returning the $O(k)$ keys with the largest estimated magnitudes (via the median estimator).

As we mentioned, in this work we are interested in the {\em adaptive setting}, where the same initialization (specified by $\boldsymbol{\rho}$) is used to sketch different inputs or to maintain a sketch as the input is updated. In its most basic form, the setting is modeled using the following game between a sketching algorithm $\sketch$ with an estimator $M$ (not necessarily \texttt{CountSketch} and the median estimator) and an \texttt{Analyst}. At the beginning, we sample the initialization randomness of $\sketch$, denoted as $\boldsymbol{\rho}$. Then, the game proceeds in $m$ rounds, where in round $q\in[m]$:
\begin{itemize}
    \item The \texttt{Analyst} chooses a vector  $\v_q\in\R^n$, which can depend, in particular, on all previous outputs of $M$.
    \item $M(\sketch_{\bro}(\v_q))$ outputs a set $H_q$ of $O(k)$ keys, which is given to the \texttt{Analyst}.
\end{itemize}

We say that $(\sketch,M)$ is {\em robust for $m$ rounds} if for any \texttt{Analyst}, with high probability, for every $q\in[m]$ it holds that $H_q$ contains all the heavy hitters of $\v_q$. The focus is on designing robust sketches of size as small as possible (as a function of $n,k,m$). We remark that it is trivial to design robust sketches with size linear with $m$, by simply duplicating a classical sketch $m$ times and using every copy of the classical sketch to answer one query. Therefore, if a sketch is only robust to a number of rounds that scales linearly with its size, then we simply say that this sketch is {\em non-robust}.

\subsection{Our Contributions}

We first briefly state our main contributions. We elaborate on these afterwards.

\begin{enumerate}
    \item We show that \texttt{CountSketch} together with the standard median estimator are non-robust. We achieve this by designing an attack that poses $O(d/b)$ queries to \texttt{CountSketch} such that, w.h.p.,\ the answer given to the last query is wrong. This constitutes the first result showing that a popular sketching algorithm is non-robust. We complement our analysis with empirical evaluation, showing that our attack is practical.
    
    \item We introduce (see Section~\ref{sec:signalign}) a novel estimator (instead of the median estimator), which we refer to as the {\em sign-alignment estimator}, that reports keys as heavy hitters based on the {\em signs} of their corresponding buckets. This new estimator is natural, and as we show, has comparable performances to the median estimator even in the non-adaptive setting. We believe that this new estimator can be of independent interest.
    
    \item We design a noisy version of our sign-alignment estimator  (utilizing techniques from differential privacy), and design a variant of \texttt{CountSketch}, which we call 
    \texttt{BucketCountSketch}, or 
    \texttt{BCountSketch} in short. We show that \texttt{BCountSketch} together with our noisy estimator are robust for $m$ rounds using sketch size roughly $\approx \sqrt{m} \cdot k$. This improves over the previous state-of-the-art robust algorithm for the heavy hitters problem of \cite{HassidimKMMS20}, which has size $\approx \sqrt{m} \cdot k^{1.5}$.
    
     We show that, in a sense, the additional ingredients from differential privacy in the robust estimator are necessary:  $\Bcountsketch$ and $\countsketch$ with a basic version of the sign-alignment estimator are non-robust. Moreover, our analysis of the robust noisy version is tight in that it can be attacked using $\Otilde(m^2)$ rounds.
    
    \item Extensions: We show that our algorithms allow for robustly reporting estimates for the weight of the identified heavy-hitters. We also refine our robust algorithms so that robustness is guaranteed for longer input sequences in some situations. In particular, this is beneficial in a streaming setting, where the input vector changes gradually and we can only account for changes in output.
\end{enumerate}

\subsubsection{Our attack on \texttt{CountSketch} with the median estimator}

We next provide a simplified description of our attack on \texttt{CountSketch} with the median estimator. We emphasize that our attack is much more general, applicable to a broader family of sketches and estimators (see Appendix~\ref{sec:attackstrategy}). Recall that on every iteration, when given a vector $\v$, we have that $M(\sketch_{\bro}(\v))$ outputs a set of $k'=O(k)$ keys. In our attack, we pose a sequence of $m=O(d/b)$ vectors $\v_1,\dots,\v_m$, where all of these vectors contain keys 1 and 2 as ``largish heavy hitters'' (of equal value), as well as a fixed set of $k'-1$ ``super heavy hitters'', say keys  $3,4,\dots,k'+1$. In addition, each $\v_q$ contains a disjoint set of keys (say $T$ keys) with random $\{\pm1\}$ values, which we refer to as a {\em random tail}. So each of these vectors contains $k'-1$ ``super heavy hitters'', two ``largish heavy hitters'', and random noise. We feed each of these vectors to \texttt{CountSketch}. As \texttt{CountSketch} reports exactly $k'$ keys, in every iteration we expect all of the ``super heavy'' elements to be reported, together with {\em one} of the ``largish heavy hitters'' (as a function of the random tail). Let us denote by $\collected\subseteq[m]$ the subset of all iterations during which key $1$ was {\em not} reported. (We refer to a tail used in iteration $q\in\collected$ as a {\em collected tail}.) The fact that key 2 was reported over key 1 during these iterations means that our random tails introduced some bias to the weight estimates of key 2 over the weight estimates of key 1. We show that, conditioned on a tail being collected, in expectation, the random tail introduces {\em negative} bias to the weight estimate of key 1, and {\em positive} bias to the weight estimate of key 2. At the end of this process, we construct a final query vector $\v_{\rm final}$ containing key 1 as a ``super heavy hitter'' and keys $2,3,\dots,k'+1$ as ''largish heavy hitters'', together with the sum of all the collected tails. We show that, w.h.p.,\ the biases we generate ``add  up'', such that key 1 {\em will not get reported as heavy} when querying $\v_{\rm final}$, despite being the dominant coordinate in $\v_{\rm final}$ by a large margin.

\begin{remark}
An important feature of our attack is that it works even when $\countsketch$ is only used to report a set of heavy keys, {\em without their estimated values}. The attack can be somewhat simplified in case estimated values are reported.  
\end{remark}

\begin{remark}
A natural question to ask regards the implications of our attack on the robustness of routine applications of $\countsketch$. We argue that elements of the attack can occur in routine settings.
The attack is "blackbox:" Only uses information from the output. It collects and combines inputs with the same heavy key, which is a natural simple feedback process. In practice this can correspond to examples with the same label or to related traffic patterns that might load the same network component. Our attack combines components of the input that contribute to a  "misclassification."  Moreover, adaptivity of the input is very lightly used: The final query $\v_{\rm final}$ on which the sketch fails is the only one that depends on prior inputs.  This suggests that an adversarial input can be constructed after simply monitoring the "normal operation" of the algorithm on unintentional inputs.  On the other hand, our attack uses "borderline" inputs in order to construct an adversarial input.  This suggests that the sketch would be more robust in settings that somehow restrict all inputs (in a deterministic way) to be far from the decision boundary -- with all keys being either very heavy or very far from being heavy.  
\end{remark}

We complement our theoretical analysis of this attack with empirical results, showing that the attack is feasible with a small number of queries.
Figure~\ref{attack:plot} reports simulation results of the attack on the median estimator.  The left plot shows the bias-to-noise ratio 
(the median value of the tail contributions scaled by its standard deviation) as a function of the number of attack rounds for the two special keys (which accumulate positive and negative bias) and another key (which remains unbiased).  The left plot visualizes the square-root relation of the bias-to-noise ratio  with the number of rounds. 
A sketch provides good estimates for a key when its weight is larger than the "noise" on its buckets that is induced by the "tail" of the vector. In this case, the key is a heavy hitter.  The attack is thus effective when the bias exceeds the noise.  
The right plot shows the number of rounds needed to achieve a specified bias-to-noise ratio (showing $\{1,4\}$) as a function of the size (number of rows $\ell=d/b$) of the sketch.  The results indicate that  $r\approx 5\cdot \texttt{BNR}^2\cdot (d/b)$ rounds are needed to obtain a vector with bias-to-noise ratio of $\texttt{BNR}$ for a sketch with $(d/b)$ rows.\footnote{We performed additional simulations (not shown) where we swept $b$ from $30$ to $300$, keeping $d/b =100$ and keeping $k'/b=3$. We observed (as expected) the same dependence of $r$. Note that the dependence should hold as  long as the parameters are in a regime where most buckets of each of the $k'$ heaviest keys have no collisions with other heavy keys.  Also note that we used a large value of $k' = b/3$, but the same dependence holds for smaller values of $k'$.}

\begin{figure}[H]
\includegraphics[width=8cm]{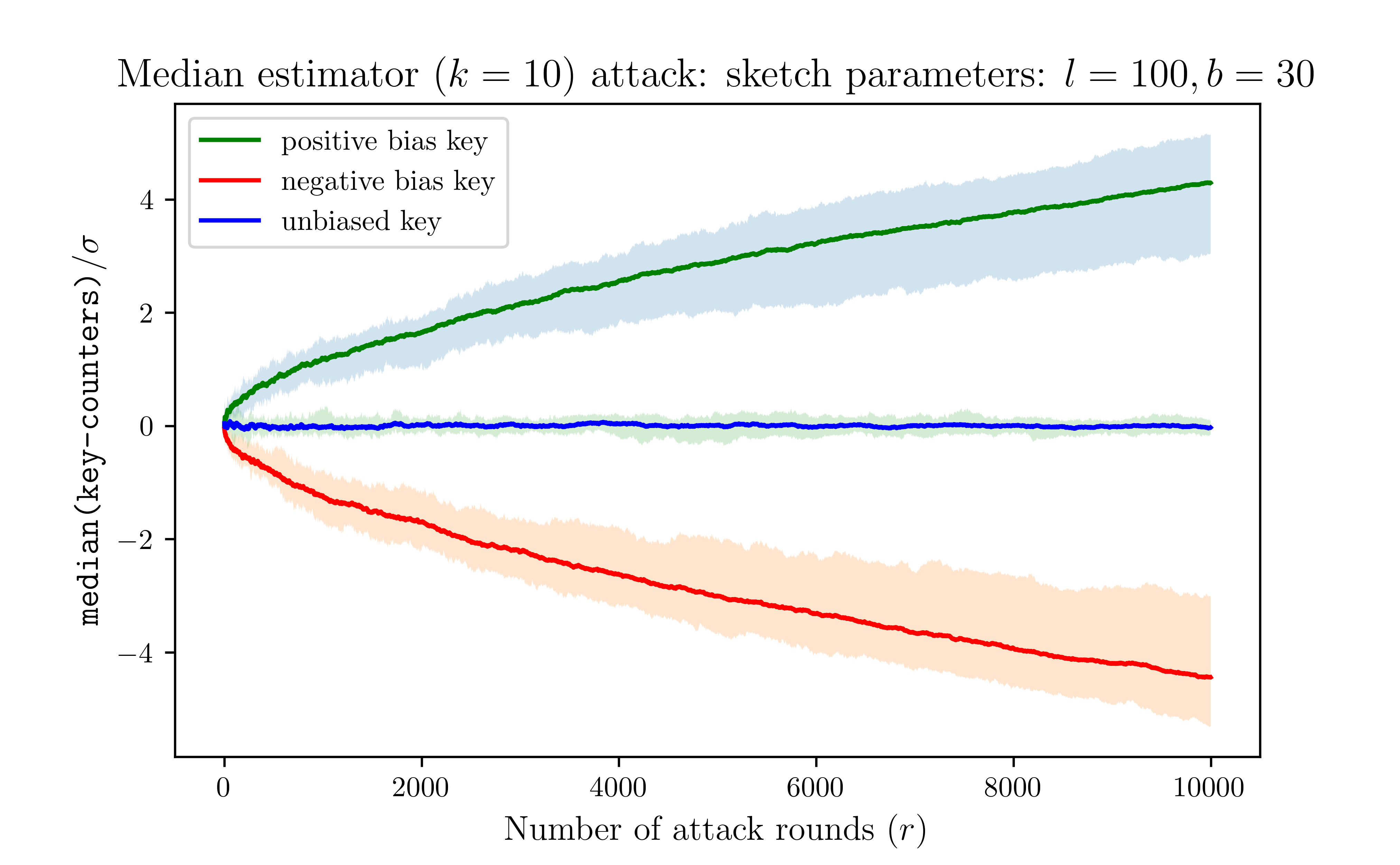}
\includegraphics[width=8cm]{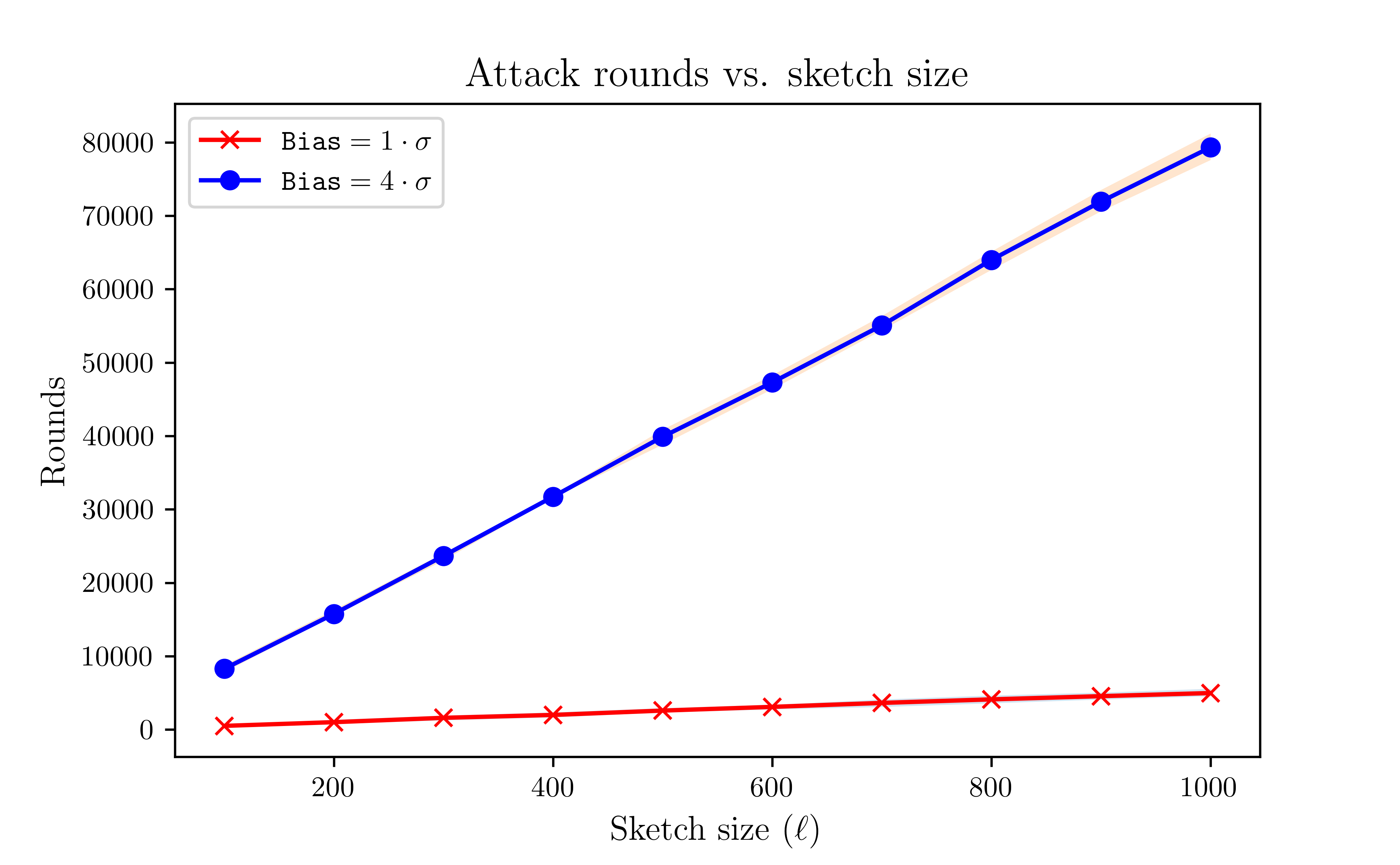}
\caption{Left: Bias-to-noise ratio for number of rounds, average of 10 simulations with different initializations (shaded region between the minimum and maximum).  Right: Attack rounds versus sketch size $\ell=d/b$ to obtain bias-to-noise ratio $\in \{1,4\}$,  averaged over $20$ and $5$ simulations respectively (shaded region between minimum and maximum).} \label{attack:plot}
\end{figure}
\subsubsection{Our new robust sketch using differential privacy}

Differential privacy \cite{dwork2006calibrating} is a mathematical definition for privacy that aims to enable statistical analyses of datasets while providing strong guarantees that individual level information does not leak. Specifically, an algorithm that analyzes data is said to be {\em differentially private} if it is insensitive 
to the addition/deletion of any single individual record in the data. Intuitively, this guarantees that whatever is learned about an individual record could also be learned without this record. Over the last few years, differential privacy has proven to be an important algorithmic notion, even when data privacy is not of concern. Particularly of interest to us is a recent line of work, starting from~\cite{DworkFHPRR15}, showing that differential privacy can be used as a tool for avoiding overfitting, which occurs when algorithms begin to ``memorize'' data rather than ``learning'' to generalize from a trend.

Recall that the difficulty in our adaptive setting arises from potential dependencies between the inputs of the algorithm and its internal randomness. Our construction builds on a technique, introduced by~\cite{HassidimKMMS20}, for using differential privacy to protect not the input data, but rather the internal randomness of algorithm. As \cite{HassidimKMMS20} showed, leveraging the ``anti-overfitting'' guarantees of differential privacy, this overcomes the difficulties that arise in the adaptive setting. Following \cite{HassidimKMMS20}, this technique was also used by \cite{AttiasCSS21,BEO21,gupta2021adaptive,BKMNSS22} for designing robust algorithms in various settings. At a high level, all of these constructions operate as follows. Let $\calA$ be a randomized algorithm that solves some task of interest {\em in the non-adaptive setting} (in our case, $\calA$ could be \texttt{CountSketch} combined with the median estimator). To design a robust version of $\calA$ do the following.

\begin{center}
\fbox{%
  \parbox{0.95\linewidth}{
\begin{center}
    {\bf \vspace{-4px}\texttt{GenericTemplate}}\\[-8px] \hrulefill 
\end{center}  
\begin{enumerate}[topsep=0pt,rightmargin=10pt]
    \item Instantiate $P$ copies of the non-adaptive algorithm $\calA$ (for some parameter $P$).
    \item For $m$ steps: 
    \begin{enumerate}[topsep=0pt,leftmargin=15pt]
        \item Feed the current input to all of the copies of $\calA$ to obtain $P$ intermediate outputs $y_1,\dots,y_P$.
        \item\label{step:GTAgg} Return a differentially private aggregation of $\{y_1,\dots,y_P\}$ as the current output.
    \end{enumerate}
\end{enumerate}
}%
}
\end{center}

That is, all the current applications of differential privacy for constructing robust algorithms operate in a black-box fashion, by aggregating the outcomes of non-robust algorithms. 
For the $\ell_2$-heavy hitters problem, as we next explain, we would need to set $P\approx\sqrt{m\cdot k}$, which introduces a blowup of $\approx\sqrt{m\cdot k}$ on top of the size of the non-adaptive algorithm. Applying this construction with \texttt{CountSketch}, which has size $\Otilde(k)$, results in a robust algorithm for $m$ rounds with size $\approx\sqrt{m}\cdot k^{1.5}$.

The reason for setting $P\approx\sqrt{m\cdot k}$ comes from known composition theorems for differential privacy, showing that, informally, we can release $T$ aggregated outputs in a differentially private manner given a dataset containing $\approx\sqrt{T}$ elements. In our context, we have $m$ rounds, during each of which we need to release $O(k)$ heavy elements. This amounts to a total of $O(m\cdot k)$ aggregations, for which we need to have $\approx\sqrt{m\cdot k}$ intermediate outputs. 

We design better algorithms by breaking free from the black-box approach outlined above. We still use differential privacy to protect the internal randomness of our algorithms, but we do so in a ``white-box'' manner, by integrating it into the sketch itself. As a warmup, consider the following noisy variant of the median estimator (to be applied on top of \texttt{CountSketch}). Given the sketch of a vector $\v$, denoted as $\sketch_{\boldsymbol{\rho}}(\boldsymbol{v}):=
\left(c_{(j,w)}\right)_{j\in[d/b],w\in[b]}$, and a coordinate $\ell$, instead of returning the actual median of $\left\{
s_j(\ell)\cdot c_{\left(j,h_j(\ell)\right)}
 \right\}_{j\in[d/b]}$, return a differentially private estimate for it. 
 
 As before, in order to release $O(m\cdot k)$ estimates throughout the execution, we would need to set $d/b\approx\sqrt{m\cdot k}$, which results in a sketch of size $\frac{d}{b}\times b\approx\sqrt{m}\cdot k^{1.5}$. So we did not gain much with this idea in terms of the sketch size compared to the \texttt{GenericTemplate}. Still, there is a conceptual improvement here. The improvement is that with the \texttt{GenericTemplate} we argued about differential privacy w.r.t.\ the intermediate outputs $y_1,\dots,y_P$, where every $y_p$ results from a different instantiation of \texttt{CountSketch}. This effectively means that in the \texttt{GenericTemplate} we needed to tune our privacy parameters so that the aggregation in Step~\ref{step:GTAgg} ``hides'' an entire copy of \texttt{CountSketch}, which includes $2d/b$ hash functions. With our warmup idea, on the other hand, when privately estimating the median of $\left\{
s_j(\ell)\cdot c_{\left(j,h_j(\ell)\right)}
 \right\}_{j\in[d/b]}$, we only need to hide every single one of these elements, which amounts to hiding only a single hash function pair $(h_j,s_j)$. (We refer to the type of object we need to ``hide'' with differential privacy as our {\em privacy unit}; so in the \texttt{GenericTemplate} the privacy unit was a complete copy of \texttt{CountSketch}, and now the privacy unit is reduced to a single hash function pair). 
 This is good because, generally speaking, protecting less with differential privacy is easier, and can potentially be done more efficiently.
 
Indeed, we obtain our positive results by ``lowering the privacy bar'' even further. Informally, we show that it is possible to work with the individual {\em buckets} as our privacy unit, rather than the individual hash functions, while still being able to leverage the generalization properties of differential privacy to argue utility in the adaptive setting. Intuitively, but inaccurately, by doing so we will have $\approx d$ elements to aggregate with differential privacy (the number of buckets), rather than only $\approx d/b\approx d/k$ elements (number of hash functions), which would allow us to answer a larger number of adaptive queries via composition arguments.
 
There are two major challenges with this approach, which we need to address.

\paragraph{First challenge.}
Recall that every key $i$ participates in $d/b$ buckets (one for every hash function). So, even if we work with the individual buckets as the privacy unit, still, when estimating the weight of key $i$ we have only $d/b$ elements to aggregate; not $d$ elements as we described it above. It is therefore not clear how to gain from working with the buckets as the privacy unit. We tackle this by conducting a more fine-tuned analysis, based on the following idea. Suppose that the current input vector is $\v$ and let $H$ denote the set of $O(k)$ keys identified to be heavy. While we indeed estimate the weight of every $i\in H$ by aggregating only $d/b$ buckets, what we show is that, on average, we need to aggregate {\em different} buckets to estimate the weight of every $i\in H$. This means that (on average) every bucket participates in very few estimations per query. Overall, every bucket participates in $O(m)$ aggregations throughout the execution, rather than $O(m \cdot k)$ as before. Using composition arguments, we now need to aggregate only $\sqrt{m}$ elements (buckets) to produce our estimates, rather than $\sqrt{m \cdot k}$ as before. So it suffices to set $d/b\gtrsim\sqrt{m}$, i.e., suffices to set $d\gtrsim\sqrt{m}\cdot b\approx\sqrt{m}\cdot k$. The analysis of this idea is delicate, and we actually are not aware of a variant of the median estimator that would make this idea go through. To overcome this issue, we propose a novel estimator, which we refer to as the {\em sign-alignment estimator}, that reports keys as heavy hitters based on the signs of their corresponding buckets. This estimator has several appealing properties that, we believe, make it of independent interest.

\paragraph{Second challenge.}
The standard generalization properties of differential privacy, which we leverage to avoid the  difficulties that arise from adaptivity, only hold for product distributions.\footnote{While there are works that studied the generalization properties of (variants of) differential privacy under non-product distributions, these works are not applicable in our setting.~\cite{bassily2016typical,kontorovich2022adaptive}} This is fine when working with the individual hash functions as the privacy unit, because the different hash functions are sampled independently. However, this is no longer true when working with the individual buckets as the privacy unit, as clearly, buckets pertaining to the same hash function are dependent. To overcome this difficulty, we propose a variant of $\countsketch$ which we call $\Bcountsketch$, that has the property that all $d$ of its buckets are independent. This variant retains the marginal distribution of the buckets in $\countsketch$, but removes dependencies.  

\subsubsection{Feasibility of alternatives} 
Deterministic algorithms are inherently robust and therefore one
approach to achieve robustness is to redesign the algorithm  to be deterministic or have deterministic components~\cite{ShiloachEven:JACM1981,GutenbergW:SODA2020}.  We note that the related $\ell_1$-heavy hitters problem on data streams has known  deterministic sketches~\cite{MisraGries:1982,NelsonNW13}, and therefore a sketch that is fully robust in adaptive settings.
For $\ell_2$-heavy hitters, however, all known designs are based on randomized linear measurements and there are known lower bounds of $\Omega(\sqrt{n})$ on the size of any deterministic sketch, even one that only supports positive updates \cite{KamathPW:CCC2021}.  In particular this means that for $\ell_2$-heavy hitters we can not hope for better robustness via a deterministic sketch.

\section{Preliminaries}\label{sec:preliminaries}

For vectors $\boldsymbol{v},\boldsymbol{u}\in\mathbb{R}^n$ we use the notation $v[i]$ for the value of the $i$th entry of the vector (which we also refer to as the $i$th key), 
$\langle \boldsymbol{v},\boldsymbol{u}\rangle = \sum_{i=1}^n v[i] u[i]$ for the inner product, and $\|\boldsymbol{v}\|_2 := (\sum_{i=1}^n v[i]^2)^{1/2}$ for the $\ell_2$ norm.

\begin{definition} (heavy hitter)
Given a vector $\boldsymbol{v}\in \mathbb{R}^n$, 
an entry $i\in[n]$ is an $\ell_2$-$k$-heavy hitter if $v[i]^2 > \frac{1}{k} \|\boldsymbol{v}_{\tail[k]} \|_2^2$.
\end{definition}

\begin{definition} (Heavy hitters problem, with and without values)
A set of entries $K\subset [n]$ is a correct solution for 
the heavy hitters problem if $|K|= O(k)$ and $K$ includes all the heavy hitters.  The solution is $\alpha$-correct for the problem with values if it includes approximate values $\hat{v}[i]$ for all $i\in K$ so that
$|\hat{v}[i] - v[i]| \leq (\sqrt{\alpha/k})\|\boldsymbol{v}_{\tail[k]}\|_2$.
\end{definition}

\subsection{$\countsketch$ and $\Bcountsketch$}

$\countsketch$ and our proposed variant $\Bcountsketch$ are 
specified by the parameters $(n,d,b)$, 
where $n$ is the dimension of input vectors, $b$ is its {\em width}, and $d$ is the size of the sketch (number of linear measurements of the input vector).

The internal randomness $\rho$ of the sketch specifies a set of $d$ 
measurement vectors $(\boldsymbol{\mu}_t)_{t\in [d]}$ where 
$\boldsymbol{\mu}_t\in \{-1,0,1\}^n$ for $t\in[d]$.
The sketch of a vector $\boldsymbol{v}\in \mathbb{R}^n$ is the set of $d$ linear measurements (which we also refer to as {\em buckets})
\[ \sketch_{\boldsymbol{\rho}}(\boldsymbol{v}) := \left(c_{t}(\boldsymbol{v}) := \langle \boldsymbol{\mu}_{t},\boldsymbol{v}  \rangle\right)_{t\in[d]}  .\]

\paragraph{\texttt{CountSketch}.}
The internal randomness specifies a set of random hash functions 
$h_r:[n]\rightarrow [b]$ ($r\in [d/b]$) with the marginals that $\forall j\in [b]$, $i\in [n]$, $\Pr[h_r(i) = j]=1/b$, and 
$s_{r}:[n]\rightarrow \{-1,1\}$ ($ r\in [d/b]$) so that $\Pr[s_{r}(i)=1]=1/2$. 
The $d$ measurement vectors are organized as
$d/b$ sets of $b$ measurements each.
$\boldsymbol{\mu}_{(r-1)\cdot b + j}$ ($r\in [d/b]$, $j\in [b]$):
\[ \mu_{(r-1)\cdot b+j}[i] := \mathbbm{1}_{h_r(i) = j} s_{r}(i) .\]
Interestingly, limited (pairwise) independence 
of the hash functions $h_r$ and $s_{r}$ suffices for the utility guarantees (stated below). Note that with $\countsketch$ the measurement vectors within each set are dependent.

\paragraph{\texttt{BCountSketch}.}
The measurement vectors are drawn i.i.d.\ from a distribution $\boldsymbol{\mu}\sim \mathcal{B}$. 
The distribution $\mathcal{B}$ is the same as that of the measurement vectors of $\countsketch$ except that dependencies are removed. The i.i.d.\ property will facilitate our analysis of robust estimation. 

Each $\boldsymbol{\mu}_t$ ($t\in [d]$) is specified by two objects: A selection hash function $h_t$ and a sign hash function $s_t$ with the following marginals:
\begin{itemize}
	\item $h_t:[n] \rightarrow \{0,1\}$ s.t. $\forall i\in [n] \Pr[h_t(i)=1] = 1/b$.
	\item $s_t:[n] \rightarrow \{+1,-1\}$ s.t. $\forall i\in [n] \Pr[s_t(i)=1] = \Pr[s_t(i)=-1] = 1/2$.
\end{itemize}
The measurement vector entries are $\mu_t[i] :=  h_t(i) s_t(i)$ ($i\in [n]$).
Our upper bound only requires limited independence (3-wise for $h_t$ and 5-wise for $s_t$, respectively). Our lower bounds hold in the stronger model of full independence.

\subsection{The median estimator}
We say that a key $i \in [n]$ {\em participates} in bucket $t\in [d]$ when 
$\mu_t[i]\not=0$. We denote by $T_i := \{t\mid \mu_t[i]\not= 0\}$ the set of buckets that $i$ participates in. Note that
with $\countsketch$ we have that $|T_i| = d/b$ since $i$ participates in exactly one bucket in each set of $b$ buckets and with $\Bcountsketch$ we have
$\E[ |T_i|]=d/b$ since $i$ participates in each bucket with probability $1/b$.  Also note that with both methods, $\boldsymbol{\mu}_t$ for $t\in T_i$ are i.i.d.

Note that for all $i\in[n]$ it holds that
$\E_{\boldsymbol{\rho}}[\mu_t[i] \cdot c_t \mid t\in T_i ] = v[i]$.
For each key
$i\in [n]$ we get a multiset of unbiased independent weak estimates of the value $v[i]$ (one for each $t\in T_i$):
\[
V(i) := \{\mu_t[i] \cdot c_t \mid \mu_t[i]\not=0\} .
\]
We use these estimates to determine if $i$ should be reported as a heavy hitter and if so, its reported estimated value. 
The classic $\countsketch$ estimator \cite{CharikarCFC:2002} uses the median of these values: $\hat{v}[i] := \textrm{median} V(i)$.  
\subsection{Utility of \texttt{CountSketch} and \texttt{BCountSketch}} \label{utilitymedian:sec}
The median estimator guarantees that for $\delta\in (0,1)$ and $d = O(b\cdot \log(1/\delta)$,
$\Pr_{\boldsymbol{\rho}\sim\mathcal{D}}\left[(v[i]-\hat{v}[i] )^2 > \frac{1}{b} \|\boldsymbol{v}_{\tail[b]} \|_2^2 \right]\leq \delta$, where 
$\boldsymbol{v}_{\tail[k]}$, 
the $k$-tail of $\boldsymbol{v}$, is a vector obtained from $\boldsymbol{v}$ by replacing its $k$ largest entries in magnitude with $0$.
The analysis extends to $\Bcountsketch$ (that has the same distribution of the independent bucket estimates except that their number is $d/b$ in expectation and not exact).  The median estimator is unbiased whereas other quantiles of $V(i)$ may not be, but importantly for our robust weight estimation,  the stated error bound holds for any {\em quantile} $\hat{v}[i]$ of $V(i)$ in a range $(1/2-\phi,1/2+\phi)$ of $V(i)$, where the constant $\phi< 1/2$ can be tuned by the constant factors of setting the sketch parameters \cite{MintonPrice:SODA2014}.

The following $\ell_\infty/\ell_2$ guarantee is obtained using a union bound ver keys \cite{CharikarCFC:2002} (for $\alpha,\delta \in (0,1/4)$ and sketch parameters $b= O(k/\alpha)$ and $d=O((k/\alpha) \log(n/\delta))$):
\begin{equation} \label{CSketchAcc:eq}
\Pr_{\boldsymbol{\rho}\sim\mathcal{D}}\left[\| \hat{\boldsymbol{v}}^{(\boldsymbol{\rho})} -\boldsymbol{v} \|_\infty^2 > \frac{\alpha}{k} \|\boldsymbol{v}_{\tail[k/\alpha]} \|_2^2\right] \leq \delta ,
\end{equation}
where 
$\hat{\boldsymbol{v}}^{(\boldsymbol{\rho})}$ is an approximation of 
$\boldsymbol{v}$ that is computed from $\sketch_{\boldsymbol{\rho}}(\boldsymbol{v})$.
For the 
heavy hitters problem, we return the set of $k(1+1/(1-2\alpha))$ keys with largest estimates, along with their estimated values.

When we have $m$ different {\em non-adaptive} inputs $(\boldsymbol{v}_q)_{q\in [m]}$, a simple union bound argument with
\eqref{CSketchAcc:eq} provides that with a sketch parameters $b=O(k/\alpha)$ and $d=O((k/\alpha) \log(nm/\delta))$: 
\begin{equation} \label{CSketchAcc2:eq}
\Pr_{\boldsymbol{\rho}\sim\mathcal{D}}\left[ \forall q\in [m],\ 
\|\hat{\boldsymbol{v}}_q^{({\boldsymbol{\rho}})} -\boldsymbol{v}_q \|_\infty^2 \leq  \frac{\alpha}{k}  \|(\boldsymbol{v}_{q})_{\tail[k]} \|_2^2\right] \geq 1-\delta\ .
\end{equation}
That is, the number of inputs for which we can guarantee utility with high probability grows {\em exponentially} with the size of the sketch. As mentioned in the introduction, we shall see that the median estimator is not robust in adaptive settings, where we can only guarantee utility for number of inputs that grows {\em linearly} with the sketch size, matching a trivial upper bound.

\section{Sign-Alignment Estimators}\label{sec:signalign}

We propose sign-alignment estimators
(with $\countsketch$ and $\Bcountsketch$)
that determine whether a key $i$ is reported as a potential heavy hitter based on the number of buckets for which the signs of $\mu_t[i] \cdot c_t$ align. 

For $\boldsymbol{\mu}\in \Supp(\mathcal{B})$, $\boldsymbol{v}\in \mathbb{R}^n$, and $i\in [n]$, we 
define the predicates
\begin{align*}
    h^+_{\boldsymbol{v},i}(\boldsymbol{\mu}) &:= \1\left\{\langle \boldsymbol{\mu},\boldsymbol{v}\rangle \mu[i]>0\right\}\\
    h^-_{\boldsymbol{v},i}(\boldsymbol{\mu}) &:= \1\left\{\langle \boldsymbol{\mu},\boldsymbol{v}\rangle \mu[i]<0\right\}\ .
\end{align*}

We show that for a key $i$ that participates in the bucket $t$, if $i$ is heavy then the sign of $v[i]$ is very likely to agree with the sign of the bucket estimate $\mu[i] c_t$ but when there are many keys that are heavier than $i$ ($i$ lies in the "tail") then such agreement is less likely.


For $\boldsymbol{v}$ and $i\in [n]$ we accordingly define the probabilities that these predicates are satisfied by $\boldsymbol{\mu}$, conditioned on $i$ participating, as
\begin{align*}
    p^+(\boldsymbol{v},i) &:= 
\Pr_{\boldsymbol{\mu}}[h^+_{\boldsymbol{v},i}(\boldsymbol{\mu}) \mid \mu[i]\not=0] = b \cdot \Pr_{\boldsymbol{\mu}}[h^+_{\boldsymbol{v},i}(\boldsymbol{\mu}) ]\\
   p^-(\boldsymbol{v},i) &:= \Pr_{\boldsymbol{\mu}}[h^+_{\boldsymbol{v},i}(\boldsymbol{\mu}) \mid \mu[i]\not=0] = b \cdot \Pr_{\boldsymbol{\mu}}[h^-_{\boldsymbol{v},i}(\boldsymbol{\mu})]\\
   p(\boldsymbol{v},i) &:=  \max\{ p^+(\boldsymbol{v},i), p^+(\boldsymbol{v},i)\}
\end{align*}

The intuition is that when $v[i]^2 \ll \frac{\alpha}{k} \|\boldsymbol{v}_{\tail[k]}\|_2^2$, 
we expect $|v(i)| \ll |c_t|$ and therefore $p^+(\boldsymbol{v},i) \approx
p^-(\boldsymbol{v},i) \approx 1/2$ and thus $p(\boldsymbol{v},i) \approx 1/2$.
When $v[i]^2 > \frac{1}{k} \|\boldsymbol{v}_{\tail[k]} \|_2^2$ we expect $v[i]\mu_t[i] \approx c_t$ and $\Pr[c_t \cdot \mu_t[i] \cdot \mathrm{sgn}(v[i])>0 \mid t\in T_i] \approx 1$ and hence $p(\boldsymbol{v},i) \approx 1$.

\begin{lemma} \label{lemma:heavyornot}
There are constants $C_a$ and $C_b$ and $1/2 < \tau_a < \tau_b < 1$ such that for all
$\boldsymbol{v}\in \mathbb{R}^n$, for all $i\in [n]$
\begin{itemize}
    \item If $v[i]^2 > \frac{C_b^2}{b}\|\boldsymbol{v}_{\tail[b/C_b^2]}\|^2_2$ then 
    $p^{\mathrm{sgn}(v[i])}(\boldsymbol{v},i)\geq \tau_b$ (and therefore
    $p(\boldsymbol{v},i)\geq \tau_b$)
    \item If $v[i]^2 \le \frac{1}{b} \| \boldsymbol{v}_{\tail[C_a\cdot b]} \|_2^2$ then $p(\boldsymbol{v}, i) \le \tau_a$.
\end{itemize}

\end{lemma}

\begin{corollary} \label{pspecs:coro}
Consider sketches with width $b = C_b^2 \cdot k$ and define
\begin{align}
    \textrm{heavy}(\boldsymbol{v}) &:=
     \{i\in[n] \mid p(\boldsymbol{v},i) \geq \tau_b\}\\
    \textrm{suspect}(\boldsymbol{v}) &:= \{i\in[n] \mid p(\boldsymbol{v},i) \geq \tau_a\}\ .
\end{align}
Then the set 
$\textrm{heavy}(\boldsymbol{v})$ contains all 
heavy hitter keys of $\boldsymbol{v}$
and
$|\textrm{suspect}(\boldsymbol{v})| \leq  (C_a+1) \cdot b = (C_a+1) \cdot C_b^2 \cdot k$.
\end{corollary}
It follows from Corollary~\ref{pspecs:coro} that a set $\texttt{K}(\boldsymbol{v})$ that includes all
$\textrm{heavy}(\boldsymbol{v})$ keys and only  $\textrm{suspect}(\boldsymbol{v})$ keys is a correct solution of the
heavy hitters problem.
Our sign-alignment estimators are specified by two components. The first component is obtaining estimates 
$\hat{p}^\sigma(\boldsymbol{v},i)$ of $p^\sigma(\boldsymbol{v},i)$ given a query $\sketch_\rho(\boldsymbol{v})$ with $b=C_b^2 k$, a key $i\in[n]$, and
$\sigma\in\{-1,+1\}$.  We shall see (Section~\ref{basicest:sec}) that simple averaging suffices for the oblivious setting but more nuanced methods (Section~\ref{robustest:sec}) are needed for robustness.
The second component is the estimator that uses these estimates to 
compute an output set $\texttt{K}(\boldsymbol{v})$. We present two methods, {\em threshold} (Section~\ref{sec:basic-sa-estimator}) for arbitrary queries and {\em stable} (Section~\ref{sec:basic-stable-sa-estimator}) for continuous reporting.

\begin{remark}
Sign-alignment estimators have the desirable property that only keys $i$ that "dominate" most of their buckets can have high alignment and thus get reported. 
This because with probability $2p(\boldsymbol{v},i)-1$, the magnitude of contribution of keys $[n]\setminus\{i\}$ to a bucket is smaller than $|v[i]|$. In particular, vectors $\boldsymbol{v}$ with no heavy keys  (empty $\textrm{suspect}(\boldsymbol{v})$) will have no reported keys. This can be a distinct advantage over estimators that simply report $O(k)$ keys with highest estimates.  
\end{remark}

\subsection{The Threshold Estimator}  \label{sec:basic-sa-estimator}
A threshold sign-alignment estimator output the set of keys:
\begin{equation} \label{basicest:eq}
    \texttt{K}(\boldsymbol{v})  = \{i\in[n] \mid \max\{\hat{p}^+(\boldsymbol{v},i), \hat{p}^-(\boldsymbol{v},i) \} \geq \tau_m \},
\end{equation}
where
$\tau_m := (\tau_a + \tau_b)/2$.  

\begin{lemma} (correctness of threshold estimators) \label{smallerrthencorrect:lemma}
If for each query vector $\boldsymbol{v}$, $i\in [n]$, and 
$\sigma\in\{+,-\}$ the estimates
$\hat{p}^\sigma(\boldsymbol{v},i)$ satisfy
\begin{equation} \label{accest:eq}
|\hat{p}^\sigma(\boldsymbol{v},i)-p^\sigma(\boldsymbol{v},i)| \leq \frac{\tau_b-\tau_a}{2} ,
\end{equation}
then the output $\texttt{K}(\boldsymbol{v})$ is correct.
\end{lemma}
\begin{proof}
\begin{align*}
i\not\in \texttt{K}(\boldsymbol{v}) \implies
\hat{p}^\sigma(\boldsymbol{v},i) &<  \tau_m \implies \\
p^\sigma(\boldsymbol{v},i) &\leq \hat{p}^\sigma(\boldsymbol{v},i) + |\hat{p}^\sigma(\boldsymbol{v},i)-p^\sigma(\boldsymbol{v},i)| < \tau_m + \frac{1}{2}(\tau_b-\tau_a) = \tau_b \implies i\not\in \texttt{heavy}(\boldsymbol{v})
\end{align*}

\begin{align*}
i\in \texttt{K}(\boldsymbol{v}) \implies
\hat{p}^\sigma(\boldsymbol{v},i) &\geq  \tau_m \implies \\
p^\sigma(\boldsymbol{v},i) &\geq \hat{p}^\sigma(\boldsymbol{v},i) - |\hat{p}^\sigma(\boldsymbol{v},i)-p^\sigma(\boldsymbol{v},i)| \geq  \tau_m - \frac{1}{2}(\tau_b-\tau_a) = \tau_a  \implies i\in \texttt{suspect}(\boldsymbol{v})
\end{align*}
\end{proof}

\subsection{The Stable Estimator} \label{sec:basic-stable-sa-estimator}

This estimator is designed for a {\em continuous reporting} version of the heavy hitters problem and is beneficial
when the input sequence is of related vectors (as in streaming).  In this case we report $\texttt{K}$ continuously and modify it as needed due to input changes. In these applications we desire {\em stability} of $\texttt{K}$, in the sense of avoiding thrashing, where a borderline key exits and re-enters $\texttt{K}$ following minor updates. We shall see that stability can significantly improve robustness guarantees, as we only need to account for {\em changes} in the reported set instead of the total size of each reported set. 
 
 Our stable estimator uses two threshold values:
 \begin{align} \label{stablethresholds:eq}
 \tau_{m_1} := \tau_a + \frac{1}{5} (\tau_b-\tau_a) \\
 \tau_{m_2} := \tau_b - \frac{1}{5} (\tau_b-\tau_a)\ .
\end{align} 
 
 A key $i\not\in \texttt{K}$ enters the reported set when
 $\max\{\hat{p}^+(\boldsymbol{v},i), \hat{p}^-(\boldsymbol{v},i) \} \geq \tau_{m_2}$.
 A key $i\in \texttt{K}$ exits the reported set when
 $\max\{\hat{p}^+(\boldsymbol{v},i), \hat{p}^-(\boldsymbol{v},i) \} < \tau_{m_1}$.

\begin{lemma} (correctness of stable estimators) \label{stablesmallerrthencorrect:lemma} 
If for all queries 
$\boldsymbol{v}$, keys $i\in [n]$, and $\sigma\in\{-1,+1\}$, our estimates satisfy
\begin{equation} 
|\hat{p}^\sigma(\boldsymbol{v},i)-p^\sigma(\boldsymbol{v},i)| \leq \frac{1}{5}(\tau_b-\tau_a)
\end{equation}
then the output of the stable estimator is correct ($\texttt{K}$ includes all $\textrm{heavy}(\boldsymbol{v})$ keys and only $\textrm{suspect}(\boldsymbol{v})$ keys). Moreover, the reporting status of a key can change only when $p(\boldsymbol{v},i)$ changes by at least $\frac{1}{5}(\tau_b-\tau_a)$.
\end{lemma}
\begin{proof}
Similar to that of Lemma~\ref{smallerrthencorrect:lemma}.
\end{proof}

\subsection{Basic estimates}  \label{basicest:sec}
The {\em basic} estimates are simple averages over buckets:
\begin{equation} \label{countagree:eq}
\hat{p}^\sigma(\boldsymbol{v},i) := \frac{b}{d} \sum_{t\in [d]} \mathbbm{1}\{\mu_t[i] \cdot c_t \cdot \sigma > 0\}\ .
\end{equation}

\begin{lemma}
In an oblivious setting
(when the buckets are an independent sample from $\mathcal{B}$ that does not depend on $\boldsymbol{v}$) we have that for any constant $\tau_\Delta$ and $\beta>0$, using
$d=O(\frac{1}{\tau_\Delta ^2} b \log(mn/\beta)$, 
\[\Pr[\max_{q\in[m],i\in[n],\sigma\in\{-1,+1\}} |p^\sigma(\boldsymbol{v},i)-\hat{p}^\sigma(\boldsymbol{v},i)|> \tau_\Delta] \leq \beta \] 
\end{lemma}
\begin{proof}
From multiplicative Chernoff bound, we obtain that 
\[\Pr[|p^\sigma(\boldsymbol{v},i)-\hat{p}^\sigma(\boldsymbol{v},i)|> \tau_\Delta] \leq 2 e^{-\frac{1}{3} \tau_\Delta^2 \frac{d}{b}} \ .\]
We obtain the claim by applying a union bound over the $2mn$ queries.
\end{proof}
It follows that in the oblivious setting the sign-alignment estimators are correct with $d=O(b \log (m n/\beta)$.  This matches the utility guarantees provided with the median estimator (Section~\ref{utilitymedian:sec}).
We shall see however that as is the case for the median estimator, our sign-alignment estimators with the basic estimates are non-robust in adaptive settings.  The robust estimators we introduce in Section~\ref{robustest:sec} use estimates $\hat{p}^\sigma(\boldsymbol{v},i)$ that are more nuanced.

\section{Robust Estimators} \label{robustest:sec}

We provide two sign-alignment estimators for 
$\Bcountsketch$ that are robust against adaptive adversaries. 
A robust version of the threshold estimator of Section~\ref{sec:basic-sa-estimator}, that is described as
Algorithm~\ref{algo:robust-count-sketch}
in  Section~\ref{robustthresholddescription:sec} (correctness proof provided in Section~\ref{robustestproofs:sec}), and a 
robust version of the stable estimator of Section~\ref{sec:basic-stable-sa-estimator} that is described as
Algorithm~\ref{algo:streaming-robust} in Appendix~\ref{streaming:sec}.

In the introduction we stated the robustness guarantees in terms of the size of the query sequence, that is, a sketch with parameters $(d,b)$ provides guarantees for all query sequences $Q=(\boldsymbol{v}_q)_{q=1}^m$ where $|Q| = \tilde{O}((d/b)^2)$.  The number of inputs is a coarse parameter that uses up the same "robustness budget" for each query $\boldsymbol{v}_q$ even when very few keys are actually reported or when there is little or no change between reported sets on consecutive inputs.
We introduce a refined property of sequences, its $\lambda$-number ($\lambda_Q$),  that accounts for smaller output sizes with the threshold estimator and only for changes in the output with the stable estimators.  We then establish robustness guarantees in terms of $\lambda_Q = \tilde{O}((d/b)^2)$.  This allows the robust handling of much longer sequences in some conditions.

The $\lambda$-number we use to analyze our
robust threshold estimator (Algorithm~\ref{algo:robust-count-sketch}) 
accounts only for potential reporting of each key, namely, the number of times it occurs in $\mathrm{suspect}(\boldsymbol{v}_q)$. This saves ``robustness budget'' on inputs  $\boldsymbol{v}_q$ with a small number of heavy keys or with no heavy keys.
\begin{definition} ($\lambda$-number of an input sequence)
For an input sequence $Q=(\boldsymbol{v}_q)_{q=1}^m$ and a key $i$, define
\begin{equation} \label{lambdai:eq}
    \lambda_{Q,i} := \sum_{\boldsymbol{v}\in Q} \mathbbm{1}\{i\in \mathrm{suspect}(\boldsymbol{v})\}
\end{equation} to be
the number of vectors $\boldsymbol{v}\in Q$ for which $i$ is in $\mathrm{suspect}(\boldsymbol{v})$.
For an input sequence $Q$, define
\begin{equation} \label{sequencelambda:eq}
\lambda_{Q} := \min\big\{ \max_i \lambda_{Q,i} ,  \frac{1}{C_a \cdot b} \sum_i  \lambda_{Q,i}  \big\}\ .
\end{equation}
\end{definition}

The $\lambda$-number we use to analyze our
stable robust estimator (Algorithm~\ref{algo:streaming-robust}) accounts 
only for {\em changes} in the output between consecutive inputs. This is particularly beneficial for streaming applications, where updates to the input are incremental and hence consecutive outputs tend to be similar. For this purpose we redefine $\lambda_{Q,i}$ to bound the number of times that the key $i$ may enter or exit the reported set when a stable estimator is used (note that the redefined value is at most twice \eqref{lambdai:eq} but can be much smaller).  We then redefine $\lambda_Q$ accordingly as in \eqref{sequencelambda:eq}. Note that the redefined values always satisfy $\lambda_Q \leq 2 |Q|$ but it is possible to have $\lambda_Q \ll |Q|$, allowing for robustness on longer streams with the same budget. The approach of accounting for changes in the output in the context of robust streaming was first proposed in~\cite{BenEliezerJWY21} and we extend their use of the term {\em flip number}.
   \begin{definition} (flip number of a key) \label{flipnumber:def}
   Consider an input sequence $Q$. We say that a key $i$ is high at step $q$ 
   if $p(\boldsymbol{v}_q,i) \geq \tau_b-\frac{2}{5}(\tau_b-\tau_a)$ and is low at step $q$ if $p(\boldsymbol{v}_q,i) \leq \tau_a + \frac{2}{5}(\tau_b-\tau_a)$.
   The {\em flip number} of a key $\lambda_{Q,i}$ is defined as the number of transitions from low to high or vice versa (skipping steps where it is neither). \end{definition}
   \begin{remark} \label{flipbound:rem}
 Consider the stable estimator when $|\hat{p}^\sigma_{\boldsymbol{v}_q,i} - p^\sigma_{\boldsymbol{v}_q,i}| \leq \frac{1}{5}(\tau_b-\tau_a)$ for all $\sigma,q,i$.  The number of times key $i$ enters or exits the reported set is at most $\lambda_{Q,i}$.
\end{remark}

Our robust estimators provide the following guarantees:
\begin{theorem} \label{Adaptive_simplified:thm}
Our robust threshold (Algorithm~\ref{algo:robust-count-sketch}) and stable (Algorithm~\ref{algo:streaming-robust}) estimators, with appropriate setting of the constants, provide the following guarantees (each for its respective definition of  $\lambda_Q$):
Let $C_3, c_1$ be appropriate constants. Let $\beta > 0$.
Consider an execution of our robust estimator with adaptive inputs $Q = (\boldsymbol{v}_1,\boldsymbol{v}_2,\dots, \boldsymbol{v}_{|Q|})$, access limit $L$,  and i.i.d\ initialization of $(\boldsymbol{\mu}_t)_{t\in [d]}$ and
\begin{equation} \label{dset:eq}
d \geq C_3  \cdot b\cdot \sqrt{L} \log(L \log(m\cdot n\cdot b\cdot L)) \log (m\cdot n\cdot L) \log(\frac{m\cdot n}{\beta})
\end{equation}
Then if $\lambda_Q \leq c_1 L$ and $|Q|\leq m$, with 
probability $1-\beta$ all outputs are correct.
\end{theorem}
Restated,  we obtain that a sketch with parameters $(n,d,b)$, with our robust estimators, provides robustness to adaptive inputs $Q$ with $\lambda_q = \Otilde((d/b)^2)$.\footnote{Importantly, to obtain this robustness guarantee we do not have to actually track $\Lambda_Q$ for our input sequence so that we can stop when once a limit is reached.  The design of our algorithm allows us to determine when the guarantees "fail" (our algorithm associates an "access count" for each sketch buckets and inactivates buckets that reach an "access limit."  The accuracy guarantee fails when too many buckets of the same key turn inactive, which is something we can track. }

\subsection{The Robust Threshold Estimator} \label{robustthresholddescription:sec}
\newcommand{\TM}{\mathrm{TM}}

 Our robust threshold estimator is provided as Algorithm~\ref{algo:robust-count-sketch}
(The constants $C_1,C_2$ will be chosen sufficiently large).
The algorithm initializes a $\ThresholdMonitor$ structure $\TM$ \cite{KaplanMS21-Sparse-Vector}
(see Algorithm~\ref{algo:threshold}) over the dataset of the $d$  measurement vectors (buckets). 
A \texttt{ThresholdMonitor} inputs a predicate that is defined over $\Supp(\mathcal{B})$ and a threshold value and tests whether a noisy count of the predicate over 
$(\boldsymbol{\mu}_t)_{t\in[d]}$ exceeds the threshold.  It has the property that the privacy budget is only charged on queries where the noisy count exceeds the threshold and only buckets on which the predicate evaluates as correct are charged. The access limit $L$ specifies how many times we can charge a bucket before it gets inactivated.
The values $\lambda_{Q,i}$ constitute upper bounds on the number of times the buckets of key $i$ get charged during the execution.  
More details on $\TM$ and the correctness proof of our algorithm are 
provided in Appendix~\ref{robustestproofs:sec}.

For each query vector, the estimator loops over all keys $i\in [n]$ and tests whether the count of the predicates $h^\sigma_{\boldsymbol{v},i}$ over active buckets, with noise added, exceeds a  threshold. If so, the key is reported and the access count for the buckets that contributed to the count is incremented.  Otherwise, the key is not reported. In Appendix~\ref{faster:sec} we show that we can make this more efficient by using a non-robust heavy-hitters sketch to exclude testing of keys that are highly unlikely to be reported. 

The robust estimator as presented only reports a set of keys $\texttt{K}$.
In Appendix~\ref{estimates:sec} we describe how weight estimates can be reported as well for $i\in\texttt{K}$, by only doubling the "robustness budget" (number of accesses to buckets of $\TM$).

\begin{algorithm2e}\caption{$\ThresholdMonitor$  \cite{KaplanMS21-Sparse-Vector}}\label{algo:threshold}
\DontPrintSemicolon
\KwIn{Database $S\in X^*$, privacy parameters $\eps, \delta$, access limit $L$, adaptive sequence of queries $(f_i,s_i,\tau_i)_i$ for predicates $f_i:X\to \{0,1\}$, $s_i\in \{+1,-1\}$, and $\tau_i \in \mathbb{N}$}
\lForEach(\tcp*[h]{Initialize counters}){$x\in S$}{$c(x) \gets 0$}
$\Delta \gets  \frac{1}{\eps} \log\left( \frac{1}{\delta} \right) \log\left( \frac{1}{\eps} \log\frac{1}{\delta} \right)$\;
\ForEach(\tcp*[h]{round $i$}){query $(f_i,\sigma_i,\tau_i)$}{
Choose $a \sim \Lap(10\Delta)$ and $b \sim \Lap(\frac{1}{\eps}\log\frac{1}{\delta})$\;
$\hat{f}_i \gets f_i(S) + a + \mathbbm{1}_{\{s_i=1\}}\cdot\min\{\Delta,b\} + \mathbbm{1}_{\{s_i=-1\}}\cdot\max\{-\Delta,b\}$ \;
\If{$\hat{f}_i \cdot s_i < \tau_i \cdot s_i$}{
    \textbf{Output} $\perp$.
    }
\Else {
    \lForEach{$x\in S$}{$c(x) \gets c(x) + f_i(x)$}
    Delete from $S$ every element $x$ such that $c(x)\ge L$\;
    \textbf{Output} $\top$.\\
    }
}
\end{algorithm2e}

\begin{algorithm2e}
    \caption{Robust Threshold $\Bcountsketch$ Estimator}
    \label{algo:robust-count-sketch}
    \DontPrintSemicolon
    \KwIn{Sketch parameters $(n,d,b)$, Access limit $L$, upper bound $m \geq  L$ on the number of queries.}
    $(\boldsymbol{\mu}_t)_{t\in [d]} \gets$ initialized $\Bcountsketch$ with parameters $(n,d,b)$\;
    $\tau_m = (\tau_a+\tau_b)/2$ \tcp{threshold }
    \SetKwProg{Init}{initialize}{:}{}
    \Init{{\rm $\ThresholdMonitor$} $\TM$}{
    $\TM.S \gets (\boldsymbol{\mu}_t)_{t\in [d]}$ \tcp{Measurement vectors} 
   $\TM.(\eps,\delta) \gets  \left(\frac{C_1}{\sqrt{L}}, \frac{C_2}{n\cdot m \cdot b \cdot L} \right)$\tcp{privacy parameters}
   $\TM.L \gets L$ \tcp{access limit}
    }
     \ForEach{query sketch $(c_t = \langle \boldsymbol{\mu}_t,\boldsymbol{v} \rangle)_{t\in [d]}$ of a vector $\boldsymbol{v}\in \mathbb{R}^n$}{
     $\texttt{K} \gets \emptyset$ \tcp*{Initialize list of output keys}
    \ForEach(\tcp*[h]{loop over keys}){ $i\in [n]$}
    {
    Define $f^+:[d] \to\{0, 1\}$, where $f^+(t) = \mathbbm{1}_{\mu_t[i] \cdot c_t > 0}$\;
    Define $f^-:[d] \to\{0, 1\}$, where $f^-(t)=\mathbbm{1}_{\mu_t[i]\cdot c_t < 0}$\;
    \If{$\TM.query(f^+,+1, \frac{d}{b}\tau_m) = \top$ {\bf or} $\TM.query(f^-,+1, \frac{d}{b}\tau_m) = \top$}{$\texttt{K}\gets \texttt{K}\cup \{i\}$} } 
    \Return{$\texttt{K}$}
    } 
\end{algorithm2e}

\section{Proofs for the robust threshold estimator} \label{robustestproofs:sec}

In this Section we provide a proof of correctness of our robust threshold estimator. 
We use $\ell = d/b$ and the constants $C_a, C_b, \tau_a, \tau_b$ as in Lemma~\ref{lemma:heavyornot}.  Let $\tau_m := (\tau_a+\tau_b)/2$ and
$\tau_\Delta := (\tau_b-\tau_a)/10$.  We establish the following.

\begin{theorem} \label{Adaptive:thm}
Let $C_1, C_2, C_3, c_1$ be appropriate constants. Let $\beta > 0$.
Consider an execution of Algorithm~\ref{algo:robust-count-sketch} with adaptive inputs $Q = (\boldsymbol{v}_1,\boldsymbol{v}_2,\dots, \boldsymbol{v}_{|Q|})$ and i.i.d\ initialization of $(\boldsymbol{\mu}_t)_{t\in [d]}$. Set
\begin{equation}
d \geq C_3  \cdot b\cdot \sqrt{L} \log(L \log(m\cdot n\cdot b\cdot L)) \log (m\cdot n\cdot L) \log(\frac{m\cdot n}{\beta})
\end{equation}
Then if $\lambda_Q \leq c_1 L$, with 
probability $1-\beta$ all outputs are correct.
\end{theorem}

\subsection{Tools from Differential Privacy} \label{DPtools:sec}
First, we need to introduce necessary tools from Differential Privacy.

\begin{theorem} (Generalization property of DP \cite{DworkFHPRR15,BassilyNSSSU:sicomp2021,FeldmanS17}) \label{theo:DP-generalization}
Let $\calA:X^d \to 2^{X}$ be an $(\xi, \eta)$-differentially private algorithm that operates on a database of size $d$ and outputs a predicate $h:X\to \{0,1\}$. Let $\mathcal{D}$ be a distribution over $X$, let $S$ be a database containing $d$ i.i.d. elements from $\mathcal{D}$, and let $h\gets \calA(S)$. Then for any $T\ge 1$ we have that
{\small
\[
\Pr_{S\sim \mathcal{D},\atop h\gets \calA(S)}\left[ e^{-2\xi} h(\calD) - \frac{1}{d} \sum_{X\in \mathcal{S}} h(X)> \frac{4}{\xi d}\log(T+1) + 2T\eta \right] < \frac{1}{T}.
\]}
\end{theorem}

\subsubsection*{The Fine-Grained Sparse Vector Technique} 
We will employ the (fine-grained) sparse vector of \cite{KaplanMS21-Sparse-Vector} described in Algorithm~\ref{algo:threshold} ($\ThresholdMonitor$).  
Algorithm $\ThresholdMonitor$ has the following utility and privacy guarantees:

\begin{theorem}[Utility guarantee \cite{KaplanMS21-Sparse-Vector}]\label{theo:monitor-utility}
Consider an execution of Algorithm {\rm $\ThresholdMonitor$} on a database $S$ and on a sequence $(f_i,s_i,\tau_i)_{i\in [r]}$ of adaptively chosen queries. 
Let $S_i$ denote the database $S$ as it is before answering the $i$-th query. With probability at least $1-\beta$, it holds that for all $i\in[r]$
\begin{itemize}
    \item If the output is $\top$ then $f_i(S_i)\cdot s_i \ge \tau_i \cdot s_i - O\left(\frac{1}{\eps}\log\frac{1}{\delta}\log(\frac{1}{\eps}\log \frac{1}{\delta})\log\frac{r}{\beta} \right)$.
    \item If the output is $\perp$ then $f_i(S_i) \cdot s_i \le \tau_i \cdot s_i + O\left(\frac{1}{\eps}\log\frac{1}{\delta} \log(\frac{1}{\eps}\log \frac{1}{\delta})\log\frac{r}{\beta} \right)$.
\end{itemize}
\end{theorem}

\begin{theorem}[Privacy guarantee \cite{KaplanMS21-Sparse-Vector}]\label{theo:monitor-privacy}
Algorithm {\rm $\ThresholdMonitor$} is 
$(O(\sqrt{L} \varepsilon + L \varepsilon^2), O(L \delta))$-differentially private.
\end{theorem}

\subsection{Proof overview of Theorem~\ref{Adaptive:thm}}
Algorithm~\ref{algo:robust-count-sketch} issues to $\TM$ $2n$ count queries for each input vector $\boldsymbol{v}_q$.  These queries correspond to  predicates $h$ of the form
$h^\sigma_{\boldsymbol{v}_q,i}$ (sign $\sigma\in \{+,-\}$, 
key $i\in[n]$ and query $\boldsymbol{v}_q$).

For each predicate $h$, the respective count of the $\TM$ is computed over buckets that are active at the time of query:
\begin{align*}
F &:= \sum_{t\in [d]} \mathbbm{1}_{(\mu_t[i]\cdot c_t\cdot \sigma > 0) \wedge \text{$t$ is active}}\\ &=  d\cdot \Pr_{t\sim [d]}[h(\boldsymbol{\mu}_t) \wedge \text{$t$ is active}]\ .
\end{align*}
The $\TM$ adds noise to $F$ to obtain $\hat{F}$. 
The $\TM$ outputs $\top$ for the query $h$ if and only if $\hat{F} \geq \frac{d}{b} \tau_m $.  A key $i$ is reported if and only if the $\TM$ output is $\top$ for at least one of its two queries. Equivalently, the inclusion of each key $i$ in the output set $\texttt{K}(\boldsymbol{v}_q)$ is determined as in
\eqref{basicest:eq} using the respective approximate values
$\hat{p}^\sigma(\boldsymbol{v}_q,i) := \frac{b}{d} \hat{F}$.

It follows from Lemma~\ref{smallerrthencorrect:lemma} that if all our noisy counts $\hat{F}$ over buckets are within $\frac{d}{b} (\tau_b-\tau_a)/2$ of their expectation over $\mathcal{B}$ then the output is correct.  We will show that this happens with probability $1-\beta$.

We will bound the "error" of $\hat{F}$ by separately bounding the contributions of different sources of error. We show that with probability $1-\beta$ the additive error is at most $\frac{5}{b} \tau_\Delta$ for all these estimates.
\begin{itemize}
    \item
    One source of error is due to $\TM$ not counting inactive buckets (those that reached the access limit $L$ by $\TM$). We introduce the notion of "useful" buckets (see Section~\ref{useful:Sec}), where usefulness is a deterministic property of the input sequence and has the property that all useful buckets remain active. We show that {\em in expectation}, for each key $i$, a $(1-\tau_\Delta)$ fraction of the buckets that a key $i$  participates in are useful.  Hence in expectation the contribution to the error is bounded by $\tau_\Delta/b$.
             \item 
    Another source of error is due to the noise added by $\TM$. We use the parameter settings and Theorem~\ref{theo:monitor-utility} to bound the maximum error over all queries by $\tau_\Delta/b$ with probability $1-\beta/2$.
    \item 
   We establish correctness under the following assumption (see Section~\ref{sec:withassumption}). We treat the query vectors $Q=(\boldsymbol{v}_q)$ as fixed and assume the buckets $(\boldsymbol{\mu}_t)$ satisfy the following:
We formulate a set $H$ of $O(mn)$ predicates over $\mathcal{B}$ that depend on the query vectors $(\boldsymbol{v}_q)$ so that all have expectation 
$\leq 1/b$ and the expected value of all these predicates is approximated by the sample $(\boldsymbol{\mu}_t)_{t\in[d]}$ to within an additive error of $\tau_\Delta/b$. The set $H$ includes all predicates $h^\sigma_{\boldsymbol{v}_q,i}$ ($\sigma\in\{-1,1\},i\in[n],q\in [|Q|]$) and also includes the usefulness predicates over buckets each key $i$ participates in.  With this assumption, the total error due to inactive buckets is bounded by $2\tau_\Delta/b$ (combining their expectation and the error). Using the assumption, the error due to estimation of $h^\sigma_{\boldsymbol{v}_q,i}$ is at most $\tau_\Delta/b$. Recall also that the error due to noise is $\tau_\Delta/b$. Combining, we get an error of $4\tau_\Delta/b$ when the assumption holds.
\item
We remove the assumption by relating (see Section~\ref{sec:obliv-to-adaptive}) the count of our predicates over buckets to its respective expectation over $\mathcal{B}$ using the generalization property of DP (Theorem~\ref{theo:DP-generalization}).  The property establishes that even though
our query vectors and hence predicates are generated in a dependent way on the sampled buckets, because they are generated in a way that preserves the privacy of the buckets,  their average over the sampled buckets still approximates well their expectation. This holds with probability $1-\beta/2$ for all predicates.  
\end{itemize}


\subsection{Useful buckets} \label{useful:Sec}

\ignore{
For an input sequence $Q=(\boldsymbol{v}_q)_{q=1}^m$ and a key $i$ we define 
\[ \lambda_{Q,i} := \sum_{\boldsymbol{v}\in Q} \mathbbm{1}_{i\in \mathrm{suspect}(\boldsymbol{v})}
\]
the number of vectors $\boldsymbol{v}\in Q$ for which $i$ is in 
$\mathrm{suspect}(\boldsymbol{v})$.
}
\begin{definition} (useful buckets)
We say that a bucket with measurement vector $\boldsymbol{\mu}$ is {\em useful} with respect to $Q$ and access limit $L$ if the total count, over vectors in $Q$, of $\mathrm{suspect}(\boldsymbol{v})$ keys that participate in the bucket is at most $L$:
\[
\textrm{useful}_{Q,L}(\boldsymbol{\mu})  := \sum_{i\in [n]} \mathbbm{1}_{\mu[i]\not=0} \cdot \lambda_{Q,i} \leq L\ .
\]
\end{definition}
The predicate $\textrm{useful}$ depends on the set $Q$ and applies to all $\boldsymbol{\mu}\in \Supp(\mathcal{B})$ whereas 
being active applies only to the $d$ buckets $(\boldsymbol{\mu}_t)$ and buckets may become inactive over time.  We can relate useful and active buckets as follows:
\begin{remark} \label{activeuseful:remark}
Consider an execution of Algorithm~\ref{algo:robust-count-sketch} with input sequence $Q$ where $\TM$ only reports keys from $\mathrm{suspect}(\boldsymbol{v}_q)$ for each query 
$\boldsymbol{v}_q\in Q$. Then all useful buckets $t\in [d]$ remain active throughout the execution.  Therefore, for any predicate $h:\Supp(\mathcal{B})$, at any point during the execution, it holds that: 
\begin{align*}
 \Pr_{t\sim [d]}[h(\boldsymbol{\mu}_t) \wedge \mathrm{useful}_{Q,L}(\boldsymbol{\mu}_t) ] &\le \Pr_{t\sim [d]}[h(\boldsymbol{\mu}_t) \wedge \text{$t$ is active}] \le \Pr_{t\sim [d]}[h(\boldsymbol{\mu}_t)]\ .
\end{align*}

(for notation convenience we use $\Pr_{t\in[d]}h(\boldsymbol{\mu}_t) = \frac{1}{d} \sum_{t\in [d]} \mathbbm{1}_{h(\boldsymbol{\mu}_t)}$.)
\end{remark}

We next characterize the set of input sequences $Q$ with which most buckets of each key remain useful.
\begin{definition} (Sequences with enough useful buckets) \label{sequsefule:def}
For $L\in\mathbb{N}$, let $\mathcal{Q}(L)$ be the set of all query sequences $Q=(\boldsymbol{v}_q)_{q\in [m]}$ that satisfy
\[\forall i\in [n],\ \Pr_{\boldsymbol{\mu}}[\neg \mathrm{useful}_{Q,L}(\boldsymbol{\mu}) \wedge \mu[i]\not=0] \leq \frac{\tau_\Delta}{b}\ .\]
\end{definition}
Recall that the probability that a key participates in a bucket is  $\Pr_{\boldsymbol{\mu}}[\mu[i]\not=0]= 1/b$.  Therefore the requirement is that for all keys $i$, a random bucket $\boldsymbol{\mu}$ in which $i$ participates is useful with probability $\geq (1-\tau_\Delta)$.

We state sufficient conditions for $Q\in \mathcal{Q}(L)$:
\begin{lemma} \label{constuseful:lemma}
There is a sufficiently small constant $c_1$ so that 
for $Q=(\boldsymbol{v}_q)_{q=1}^m$, 
$\lambda_Q \leq c_1 \cdot L \implies Q\in \mathcal{Q}(L)$.  Note that in particular this means that
$|Q| \leq c_1 \cdot L \implies Q\in \mathcal{Q}(L)$.
\ignore{
\begin{enumerate}[(i)]
    \item $\max_{i\in [n]} \lambda_{Q,i} \leq c_1 L$
    \item $\sum_{i \in [n]} \lambda_{Q,i} \leq c_1 \cdot b \cdot L$
\end{enumerate}
\ignore{
\begin{enumerate}[(i)]
    \item $\max_{i\in [n]} \sum_{\boldsymbol{v}_q \in Q} \mathbbm{1}_{i\in \mathrm{suspect}(\boldsymbol{v}_q)} \leq c_1 L$
    \item $\sum_{\boldsymbol{v}_q \in Q} |\mathrm{suspect}(\boldsymbol{v}_q)|\leq c_1 \cdot b \cdot L$
\end{enumerate}
}
then $Q\in \mathcal{Q}(L)$.  Note that $|Q|\leq \frac{c_1}{\max\{1,C_a\}} \cdot L$ implies these conditions and hence $\implies Q\in \mathcal{Q}(L)$.
}
\end{lemma}
\begin{proof}

Fixing a sequence $Q$ with $\lambda_Q \leq c_1 L$, for every $i\in [n]$, we argue as follows. Since $\Pr_{\boldsymbol{\mu}}[\mu[i]\ne 0] = 1/b$, it is sufficient to show that
\[
\Pr_{\boldsymbol{\mu}}[\lnot \mathrm{useful}_Q(\boldsymbol{\mu}) \mid \mu[i] \ne 0] \le \tau_\Delta.
\]
Consider randomly drawing a bucket $\boldsymbol{\mu}$. Define a random variable 
\[
X(\boldsymbol{\mu}) = \sum_{q\in [m]}  |\{  i\in [n] \mid \mu[i]\not=0 \wedge  i\in \mathrm{suspect}(\boldsymbol{v}_q) \}|.
\]
For every $j\in [n] \setminus \{i\}$, we have
\[
\Pr_{\boldsymbol{\mu}}[\mu[j]\ne 0\mid \mu[i]\ne 0] = \frac{1}{b}.
\]
Then, by Condition (i), (ii) and linearity of expectation, we have
\[
\begin{aligned}
\Ex_{\boldsymbol \mu}[ X(\boldsymbol{\mu}) \mid \mu[i] \ne 0 ] 
&\le \lambda_{Q,i} + \frac{1}{b} \sum_{j\in [n]\setminus \{i\}} \lambda_{Q,j}  \\
&\le c_1L + \frac{1}{C_a } c_1 L
&\le 2c_1 L.
\end{aligned}
\]
Finally, by Markov's inequality, we have
\begin{align*}
\Pr_{\boldsymbol{\mu}}[\lnot \mathrm{useful}_Q(\boldsymbol{\mu}) \mid \mu[i] \ne 0] &\le \Pr_{\boldsymbol{\mu}}[ X(\mu) > L \mid \mu[i] \ne 0 ] \le 2c_1.
\end{align*}
Hence, choosing $c_1 < \frac{\tau_{\Delta}}{2}$ suffices.\footnote{The constants can be improved using Chebyshev inequality and pairwise independence of two keys mapping into the same bucket.}
\end{proof}

\subsection{Correctness with assumption}\label{sec:withassumption}

Our choice of $\TM.(\varepsilon,\delta)$ and the sketch size $d$ allow us to provide a high $(1-\beta)$ probability bound on the maximum amount of noise in all $r=2\cdot m\cdot n$ queries issued to the $\TM$:
\begin{remark} (utility) \label{utilityremark}
There are values for $C_1,C_2$ in 
the algorithm and 
$C_3$ in \eqref{dset:eq} so that 
the utility guarantee of 
Theorem~\ref{theo:monitor-utility} provides
\begin{equation} \label{utilityacc:eq}
O\left(\frac{1}{\eps}\log\frac{1}{\delta} \log(\frac{1}{\eps}\log\frac{1}{\delta}) \log\frac{2 m n}{\beta} \right) \leq \tau_\Delta \frac{d}{b}\ .
\end{equation}
\end{remark}

We establish correctness under the simplifying assumption that sampled (sketch) buckets provide good estimates for the expected counts of our predicates.

\begin{lemma} \label{oblivcorrect:lemma}
Consider an execution as follows of Algorithm~\ref{algo:robust-count-sketch} where constants are set to provide the utility guarantee \eqref{utilityacc:eq}.
Fix the buckets $(\boldsymbol{\mu}_t)_{t\in [d]}$.
Set $\TM.L$ and fix the input sequence $Q=(\boldsymbol{v}_q)_{q\in[m]} \in \mathcal{Q}(L)$.

Consider the following set of $n(4m+1)$ predicates $h:\Supp(\mathcal{B})\to \{0,1\}$:
\begin{align*}
H =& \{ h^\pm_{\boldsymbol{v}_q,i}\}_{q\in [m], i\in [n]} 
\cup
\{h^\pm_{\boldsymbol{v}_q,i}\wedge \mathrm{useful}_Q \}_{q\in [m], i\in [n]} 
\cup \left\{ \neg\mathrm{useful}_Q(\boldsymbol{\mu})\wedge \mu[i]\not=0 \right\}_{i\in [n]} \ .
\end{align*}
If for all $h\in H$,  
it holds that the average of $\mathbbm{1}_h$ over the buckets $(\boldsymbol{\mu}_t)_{t\in[d]}$ approximate well its expectation over $\boldsymbol{\mu}\sim \mathcal{B}$:
\begin{equation} \label{assumption:eq}
     \left| \Pr_{t \sim [d]} h(\boldsymbol{\mu_t}) -  \Pr_{\boldsymbol{\mu} \sim B} h(\boldsymbol{\mu}) \right| \leq \frac{1}{b} \tau_\Delta\ .
\end{equation}
Then Algorithm~\ref{algo:robust-count-sketch} is correct with probability $1-\beta$, that is, \[\forall q\in [m],\ \mathrm{heavy}(\boldsymbol{v}_q) \subseteq \texttt{K}_q \subseteq 
\mathrm{suspect}(\boldsymbol{v}_q)\ .\]
\end{lemma}
\begin{proof}

Algorithm~\ref{algo:robust-count-sketch} issues to $\TM$ $2n$ count queries for each input vector $\boldsymbol{v}_q$.  These queries correspond to  predicates $h$ of the form
$h^\sigma_{\boldsymbol{v}_q,i}$ (sign $\sigma\in \{+,-\}$, 
key $i\in[n]$ and query $\boldsymbol{v}_q$).

For each predicate $h$, the respective count of the $\TM$ is computed over buckets that are active at the time of query:
\[F := d\cdot \Pr_{t\sim [d]}[h(\boldsymbol{\mu}_t) \wedge \text{$t$ is active}]\ .\]
(recall that the set of active buckets may decrease over time).
The $\TM$ adds noise to $F$ to obtain $\hat{F}$.  It then outputs $\top$ for the query $h$ if and only if $\hat{F} \geq \frac{d}{b} \tau_m $.

From Remark~\ref{utilityremark}, with probability $1-\beta/2$, $\TM$ satisfies the utility guarantee for all $2mn$ queries, that is 
\begin{equation} \label{Fgap:eq}
    |\hat{F} -F| \leq \frac{d}{b}\tau_\Delta\ .
\end{equation}

We next bound the maximum over keys $i$ of the number of buckets amongst 
$(\boldsymbol{\mu}_t)_{t\in [d]}$ that $i$ participates in and are not useful.
\begin{align}
\Pr_{t\sim [d]} & [\neg \mathrm{useful}_Q(\boldsymbol{\mu}) \wedge \mu[i]\not=0]\nonumber\\ &\leq \frac{\tau_\Delta}{b} + \Pr_{\boldsymbol{\mu}}[\neg \mathrm{useful}_Q(\boldsymbol{\mu}) \wedge \mu[i]\not=0] 
\nonumber\\
&\leq 2\frac{\tau_\Delta}{b}, 
\label{useless:eq} 
\end{align}
where the first inequality is by our assumption on the predicate
$\neg \mathrm{useful}_Q(\boldsymbol{\mu}) \wedge \mu[i]\not=0$
and the second inequality follows from $Q\in \mathcal{Q}(L)$ (by Definition~\ref{sequsefule:def}).
Using logic, for $t\in [d]$:
\begin{align*}
  &(h(\boldsymbol{\mu}_t) \wedge \mathrm{useful}_Q(\boldsymbol{\mu}_t)) \implies h(\boldsymbol{\mu}_t) \\[1em]
&h(\boldsymbol{\mu}_t) \implies (h(\boldsymbol{\mu}_t) \wedge \mathrm{useful}_Q(\boldsymbol{\mu}_t)) 
\vee 
(\neg \mathrm{useful}_Q(\boldsymbol{\mu_t}) \wedge \mu[i]\not=0)
\end{align*}

Therefore using \eqref{useless:eq} we get
\begin{align} 
\Pr_{t\sim [d]}[h(\boldsymbol{\mu}_t)]-\frac{2}{b}\tau_\Delta &\leq
\Pr_{t\sim [d]}[h(\boldsymbol{\mu}_t) \wedge \mathrm{useful}_Q(\boldsymbol{\mu}_t)] 
\leq \Pr_{t\sim [d]}[h(\boldsymbol{\mu}_t)] \label{usefulexp:eq}
\end{align}

We now use a proof by induction over queries issued to $\TM$ to establish that the output for each query vector $\boldsymbol{v}$ and key $i$ is correct.

Suppose the output was correct until the current query. 
Then using Remark~\ref{activeuseful:remark}, the set of current active buckets is a superset of the set of useful buckets for $Q$ and a subset of all buckets.  Therefore,  
\begin{align} 
\Pr_{t\sim [d]}[h(\boldsymbol{\mu}_t) \wedge \mathrm{useful}_Q(\boldsymbol{\mu}_t)] &\leq 
\Pr_{t\sim [d]}[h(\boldsymbol{\mu}_t) \wedge \text{$t$ is active}] 
\leq
\Pr_{t\sim [d]}[h(\boldsymbol{\mu}_t)] \label{activeuseful:eq} 
\end{align}

Combining \eqref{usefulexp:eq} and \eqref{activeuseful:eq}  we obtain
\begin{align*}
 \Pr_{t\sim [d]}[h(\boldsymbol{\mu}_t)]-\frac{2}{b}\tau_\Delta &\leq    \Pr_{t\sim [d]}[h(\boldsymbol{\mu}_t) \wedge \text{$t$ is active}]
 \leq
 \Pr_{t\sim [d]}[h(\boldsymbol{\mu}_t)]
\end{align*}
\begin{align}
 \implies \left|\Pr_{t\sim [d]}[h(\boldsymbol{\mu}_t)]- \Pr_{t\sim [d]}[h(\boldsymbol{\mu}_t) \wedge \text{$t$ is active}] \right| \leq \frac{2}{b}\tau_\Delta \label{activetoall:eq}
\end{align}

From our assumption on the predicate $h$ we have
\begin{equation} \label{sampletogen:eq}
\left|\Pr_{t\sim [d]}[h(\boldsymbol{\mu}_t)- \Pr_{\boldsymbol{\mu}\sim \mathcal{B}}[h(\boldsymbol{\mu})]\right| \leq \frac{1}{b}\tau_\Delta
\end{equation}

Therefore using \eqref{Fgap:eq}, \eqref{activetoall:eq}, and \eqref{sampletogen:eq}:
\begin{align*}
&\left|\frac{1}{d}\hat{F}-
  \Pr_{\boldsymbol{\mu}\sim \mathcal{B}}[h(\boldsymbol{\mu})] \right| \\
  &\leq
 \left|\frac{1}{d}\hat{F}-  \Pr_{t\sim [d]}[h(\boldsymbol{\mu}_t) \wedge \text{$t$ is active}] \right| \\
 &+
    \left|\Pr_{t\sim [d]}[h(\boldsymbol{\mu}_t) \wedge \text{$t$ is active}]-
  \Pr_{t\sim [d]}[h(\boldsymbol{\mu}_t)] \right| \\
   &+ \left|\Pr_{t\sim [d]}[h(\boldsymbol{\mu}_t)- \Pr_{\boldsymbol{\mu}\sim \mathcal{B}}[h(\boldsymbol{\mu})]\right|\\ &\leq \frac{4}{b}\tau_\Delta
 \end{align*}
Therefore the conditions of Lemma~\ref{smallerrthencorrect:lemma} hold and the reporting is correct: Only keys in $\mathrm{suspect}$ can be reported and all keys in $\mathrm{heavy}$ are reported.
This finishes the induction step proof.
\end{proof}

In the oblivious setting, when the buckets are drawn i.i.d.\ and do not depend on $Q$, we can guarantee that the assumption holds with
probability $1-\beta$ by choosing $d=O(\log(n m/\beta))$ and applying Chernoff and union bound over predicates. In the following we analyze the adaptive setting.  

\subsection{Removing the assumption}\label{sec:obliv-to-adaptive}
We now establish that the assumption holds with probability at least $1-\beta/2$ using the generalization property of DP.
We show that with an appropriate choice of constants and $Q\in \mathcal{Q}(L)$, all predicates in $H$ are estimated well from the sample:
\begin{lemma}
Consider an execution of Algorithm~\ref{algo:robust-count-sketch} with adaptive inputs $Q$ and i.i.d\ $(\boldsymbol{\mu}_t)_{t\in [d]}$.
We show that with appropriate setting of constants and setting
the dimension as in \eqref{dset:eq}, then as long as the adaptive query sequence
satisfies
$Q\in\mathcal{Q}(L)$ and $|Q|\leq m$
then with probability $1-\beta/2$, for all 
the $n(4m+1)$ predicates $h\in H$, 
\eqref{assumption:eq} holds.
\end{lemma}
\begin{proof}
We apply Theorem~\ref{theo:DP-generalization}.
We choose the parameter in the statement of the theorem to be $T \geq 2n(4m+1)/\beta$.  With this setting and a union bound we get the bound with probability $1-\beta/2$ for all our $n(4m+1)$ predicates.

We treat Algorithm~\ref{algo:robust-count-sketch} and the "adversary" issuing the adaptive queries based on prior outputs as a single algorithm that operates on the measurement vectors $(\boldsymbol{\mu}_t)_{t\in[d]}$ and outputs queries $Q=(\boldsymbol{v}_q)$.  The predicates in $H$ are all computed from $Q$.  

The algorithm provides $(\xi,\eta)$-differential privacy. From post processing, all the predicates $H$ also preserve the privacy of $(\boldsymbol{\mu}_t)_{t\in[d]}$.  Using Theorem~\ref{theo:monitor-privacy}, we have
\begin{align*}
    \xi &= O(\sqrt{L}\varepsilon + L \varepsilon^2)\\
    \eta &= O(L \delta).
\end{align*}

Recall that $\varepsilon = C_1/\sqrt{L}$.
We choose the constant $C_1$ to be at most its value from Remark~\ref{utilityremark} and as needed to set $\xi \leq \tau_\Delta/6$.

Recall that $\delta = C_2/(m\cdot n\cdot b\cdot L)$.  We choose $C_2$ similarly to provide $\eta \leq \tau_\Delta/(6\cdot T \cdot b)$.

To summarize we have
\begin{align*}
    T &= \frac{2n(4m+1)}{\beta}\\
    \eta &\leq \frac{\tau_\Delta}{6\cdot T\cdot b}\\
    \xi &\leq \frac{\tau_\Delta}{6}\\
    d &\geq \frac{12 b \log(T+1)}{\xi \tau_\Delta}
\end{align*}

We now express the additive error on the estimate provided in Theorem~\ref{theo:DP-generalization}.
The expected value of all our predicates is at most $1/b$.  The additive contribution due to the multiplicative error term is therefore at most $(1-e^{-2\xi})/b \approx 2\xi/b$.    The additional additive terms 
are $\frac{4}{\xi d}\log(T+1)$ and $2\cdot T\cdot\eta$.  
We can substitute our choices of $T,d,\xi,\eta$ and obtain that each of these three terms is at most $\frac{1}{3b} \tau_\Delta$.  Therefore the total additive error is $\frac{1}{b} \tau_\Delta$.
\end{proof}

\section{Robust Weight Estimates} \label{estimates:sec}
Our robust estimator (Algorithm~\ref{algo:robust-count-sketch}) outputs a  set $\texttt{K}$ of $O(k)$ keys that includes all heavy hitters. In the following, we show that it is possible to also provide estimates for the weight of each key in $\texttt{K}$ within the same asymptotic robustness guarantees.

Here we make a mild assumption for the input query vector $\boldsymbol{v}\in \mathbb{R}^n$. That is, we assume each entry of $\boldsymbol{v}$ is an integer of magnitude at most polynomial in $n$ (say, $n^c$ for a constant $c$).

One approach to provide such weight estimates is to use a private median algorithm (based on the median/quantile estimator reviewed in the preliminaries Section~\ref{utilitymedian:sec}). 
There are suitable off-the-shelf private approximate median algorithms \cite{DBLP:journals/toc/BeimelNS16, DBLP:conf/focs/BunNSV15, DBLP:conf/stoc/BunDRS18, DBLP:conf/colt/KaplanLMNS20}. However, this approach constitutes a different component that does not fit seamlessly in the framework of  Algorithm~\ref{algo:robust-count-sketch}.
Instead, we design $\texttt{WeightEstimator}$ (Algorithm~\ref{algo:weight-esitmator}) that implements a private median computation by accessing the buckets only through calls to the same threshold monitor $\TM$ that is initialized in Algorithm~\ref{algo:robust-count-sketch}.

\begin{algorithm2e}
    \caption{\texttt{WeightEstimator}}
    \label{algo:weight-esitmator}
    \DontPrintSemicolon
    \KwIn{
        A query vector $\boldsymbol{v}\in \mathbb{R}^n$, an index $i\in [n]$. \;
        $(\boldsymbol{\mu}_t)_{t\in [d]} \gets$ initialized $\Bcountsketch$ with parameters $(n,d,b)$\;
        $\TM\gets$ a $\ThresholdMonitor$ based on $(\boldsymbol{\mu}_t)_{t\in [d]}$.
    }
    \SetKwProg{Init}{initialize}{:}{}
    \For{$w = -n^c, \dots, 0, 1, \dots, n^c$}{
        Define $f_{\le w} : [d] \to \{0, 1\} $ where $f_{\le w}(t) = \mathbbm{1}\{\mu_t[i] \ne 0 \land c_t\cdot \mu_t[i] \le {w}\}$ \;
        \If{TM.$query(f_{\le w}, +1, \frac{d}{b}\tau_{\rm tr}) = \top$}{
            \Return{$\tilde{v}_i = w$} \;
        }
    }
    \Return{$\tilde{v}_i = n^c$} \;
\end{algorithm2e}

Algorithm~\ref{algo:weight-esitmator}
uses the following parameters:  $\tau_{\rm down} = \frac{1}{10}$, $\tau_{\rm up} = \frac{9}{10}$, and $\tau_{\rm tr} =\frac{\tau_{\rm down} + \tau_{\rm up}}{2}$.
Note that $\TM$ reports at most one ``$\top$'' for each key in $\texttt{K}$. Therefore, we can still get the same DP guarantee (by incurring a constant multiplicative factor, because we may need to have the ``access limit'' threshold $L$ in the threshold monitor doubled\footnote{Check Algorithm~\ref{algo:threshold} for the definition of $L$.}).

\paragraph*{Efficiency.} Implementing Algorithm~\ref{algo:weight-esitmator} as described would be highly inefficient: it takes $\Theta(n^c)$ time to estimate a single key $v_i$. In the following, we describe an $\Otilde(d)$ time implementation which has identical output distribution as Algorithm~\ref{algo:weight-esitmator}.

Our efficient implementation is based on a simple observation: let $w,w+1 \in [-n^c, n^c]$ be two consecutive integers. If for every $\boldsymbol{\mu}_t$ it holds that $\mu_t[i]\ne w$, then we can conclude that $f_{\le w}$ and $f_{\le w+1}$ are the same predicate. Therefore, when querying the threshold monitor with either $f_{\le w}$ or $f_{\le w+1}$, we have the same probability of getting $\top$. The same can be said for a range of potential weights: let $w_l<w_r$ be two integers. If for every $\boldsymbol{\mu}_i$ it holds that $\boldsymbol{\mu}_i \notin [w_l, w_r - 1]$, then $(f_{\le w})$ for every $w\in [w_l, w_r - 1]$ refers to the same predicate, which means we have the same probability of getting $\perp$ when querying with $f_{\le w}$ for every $w\in [w_l, w_r-1]$. 

We are ready to describe the implementation. 

\begin{enumerate}
    \item By the observation above, we know that the buckets $(\boldsymbol{\mu}_t)_{t\in [d]}$ partition the answer space $[-n^c, n^c]$ into $d+1$ intervals ($d/b+1$ in expectation), where the same query predicate is used inside each interval. We scan these intervals from left to right.
    \item For an interval $[L, R]$, we count the number of buckets such that $\mu_i[t] \le L$ and calculate the probability that the threshold monitor reports $\top$ on any query in this range. Let $p$ denote the probability. Conditioning on Algorithm~\ref{algo:weight-esitmator} not returning before $w = L$, the probability that Algorithm~\ref{algo:weight-esitmator} does not return an estimation before moving $w$ to $R+1$ is $(1-p)^{R-L}$. With probability $(1-p)^{R-L}$, we proceed to the next interval.
    \item Otherwise, we simulate the case that the algorithm returns an integer in the range $[L, R]$. We first calculate a normalizer $P = \sum_{w=1}^{R-L+1} (1-p)^{i-1}\cdot p$. Conditioning on Algorithm~\ref{algo:weight-esitmator} returns an integer in $[L, R]$, it returns $w\in [L, R]$ with probability exactly $\frac{(1-p)^{w-L}p}{P}$. Therefore, we sample a real $r$ uniformly in $[0, P]$, and find the smallest $w\in [L, R]$ such that $r < \sum_{j=L}^{w} (1-p)^{j-L}\cdot (p)$. We output $w$.
\end{enumerate}

It is straightforward to verify that this implementation has the same output distribution as Algorithm~\ref{algo:weight-esitmator}.

\paragraph*{Correctness.} We analyze the correctness and privacy guarantee of Algorithm~\ref{algo:weight-esitmator} now. We observe that $\texttt{WeightEstimator}$ (Algorithm~\ref{algo:weight-esitmator}) fits well in the framework of Algorithm~\ref{algo:robust-count-sketch}. That is, the weight estimator only accesses the buckets through  TM (the threshold monitor). And TM reports at most one ``$\top$'' for each key in $\texttt{K}$. Therefore, we can still get the same DP guarantee (by incurring a constant multiplicative factor, because we may need to have the ``access limit'' threshold $L$ in the threshold monitor doubled).

Next, we will show that Algorithm~\ref{algo:weight-esitmator} returns an estimate of $v[i]$ within additive error $\pm \frac{1}{\sqrt{k}} \| \boldsymbol{v}_{\tail[k]} \|_2$ with high probability. 

Let $\Gamma = O\left( \frac{1}{\eps} \log\frac{1}{\delta} \log\left( \frac{1}{\eps} \log \frac{1}{\delta} \right) \log \frac{r}{\beta} \right)$ where the constant in the big-Oh is the constant in Theorem~\ref{theo:monitor-utility}. Then, by Theorem~\ref{theo:monitor-utility}, with probability at least $1-\beta$, Algorithm~\ref{algo:weight-esitmator} returns an estimation $\tilde{v}_i$ such that:
\begin{itemize}
    \item There are at least $\frac{d}{b}\tau_{\rm tr} - \Gamma$ alive buckets $\mu_t$ where $\mu_t[i] \ne 0, \mu_t[i]\cdot c_t \le \tilde{v}_t$.
    \item There are at most $\frac{d}{b}\tau_{\rm tr} + \Gamma$ alive buckets $\mu_t$ where $\mu_t[i] \ne 0, \mu_t[i]\cdot c_t \le \tilde{v}_{t} -1$.
\end{itemize}
In the following, we use $\calE_{1}$ to denote the event that both of the two conditions above hold.

Call a bucket $\boldsymbol{\mu}_t$ \emph{good} for the estimation of $v[i]$, if both of the following hold:
\begin{enumerate}
\item $\boldsymbol{\mu}_t$ is alive by the time we perform the estimation;
\item $\mu_t[i] \ne 0$ and $\mu_t[i]\cdot c_t \in v[i]\pm \frac{1}{\sqrt{k}} \| \boldsymbol{v}_{\tail[k]}\|_2$.
\end{enumerate}
Also, we call a bucket $\boldsymbol{\mu}_t$ \emph{bad} if $\boldsymbol{\mu}_t$ is alive, $\mu_t[i] \ne 0$ and $\mu_t[i]\cdot c_t \notin v[i] \pm \frac{1}{\sqrt{k}} \| \boldsymbol{v}_{\tail[k]}\|_2$. We prove the following lemma.

\begin{lemma}\label{lemma:good-estimation}
Consider the execution of Algorithm~\ref{algo:weight-esitmator} for the vector $\boldsymbol{v}$ and index $i$. Suppose all of the following hold:
\begin{itemize}
    \item The event $\calE_1$.
    \item The number of good buckets is more than $\frac{d}{b}\tau_{\rm tr} + \Gamma$.
    \item The number of bad buckets is less than $\frac{d}{b}\tau_{\rm tr} - \Gamma$.
\end{itemize}
Then, Algorithm~\ref{algo:weight-esitmator} returns an estimation of $v[i]$ within additive error $\frac{1}{\sqrt{k}} \| \boldsymbol{v}_{\tail[k]} \|_2$.
\end{lemma}

\begin{proof}
Define the multi-set $S_i = \{\boldsymbol{\mu}_t[i]\cdot c_t: \mu_t \text{ is alive} \land \mu_t[i] \ne 0 \}$. Call an element $x\in S_i$ good if $x\in v[i]\pm \frac{1}{\sqrt{k}} \| v_{\tail[k]} \|_2$. Let $n_{\rm good}$ denote the number of good elements in $S_i$. We call all other elements bad and let $n_{\rm bad}$ denote the number of such elements. We observe that $n_{\rm good}$ is equal to the number of good buckets and $n_{\rm bad}$ is equal to the number of bad buckets.

Let $\tilde{v}_i$ denote the output of the private median algorithm. Suppose $\tilde{v}_i$ is larger than or equal to $n_1$ elements in $S_i$. Since $\calE_1$ holds, we have $n_1 \ge \frac{d}{b}\tau_{\rm tr} - \Gamma$. Because $n_{\rm bad} < \frac{d}{b}\tau_{\rm tr} - \Gamma$, we know those $n_1$ elements cannot be all bad. Namely, $\tilde{v}_i$ is no less than at least one good element. 

Now, suppose $\tilde{v}_i - 1$ is larger than or equal to $n_2$ elements in $S_i$. Since $\calE_1$ holds, we have $n_2 \le \frac{d}{b}\tau_{\rm tr} + \Gamma$. Because $n_{\rm good} > \frac{d}{b}\tau_{\rm tr} + \Gamma$, we know those $n_2$ elements do not include all good elements. Namely, $\tilde{v}_i - 1$ is less than at least one good element.

Combining two observations, we conclude that $\tilde{v}_i$ is sandwiched between two good elements, which means $\tilde{v}_i \in v[i]\pm \frac{1}{\sqrt{k}} \| \boldsymbol{v}_{\tail[k]} \|$ as desired.
\end{proof}

Now let us consider the oblivious setting. In this setting, we can fix an input sequence $Q$ beforehand and analyze the behavior of our algorithm on it. In the next lemma, we will show that Algorithm~\ref{algo:weight-esitmator} is accurate, so long as the number of good/bad buckets is close to its expectation (assuming the buckets are independent to the input $Q$).

\begin{lemma} \label{obliv-estimate-correct:lemma}
Consider an execution Algorithm~\ref{algo:robust-count-sketch} where it also runs Algorithm~\ref{algo:weight-esitmator} to estimate the weight of $v[i]$ whenever it includes a key $i$ in $K$. Fix the input sequence $Q=(\boldsymbol{v}_q)_{q\in[m]} \in \mathcal{Q}(L)$ and consider one query vector $\boldsymbol{v}\in \mathcal{Q}(L)$.
Consider the following set of $2n$ predicates $h:\Supp(\mathcal{B})\to \{0, 1\}$:
\begin{align*}
H =& \{ \mathbbm{1}_{\mu[i]\ne 0 \land \mu[i]\cdot c\in v[i] \pm \frac{1}{\sqrt{k}} \|\boldsymbol{v}_{\tail[k]}\|_2 \land \mathrm{useful}_Q} \}_{i\in [n]}\\
&\cup \{ \mathbbm{1}_{\mu[i]\ne 0\land \mu[i]\cdot c\notin v[i] \pm \frac{1}{\sqrt{k}} \|\boldsymbol{v}_{\tail[k]}\|_2} \}_{i\in [n]} \ .
\end{align*}
If for all $h\in H$,  
it holds that the average of $h$ over the buckets $(\boldsymbol{\mu}_t)_{t\in[d]}$ approximate well its expectation over $\boldsymbol{\mu}\sim \mathcal{B}$:
\begin{equation} \label{assumption:close}
     \left| \Pr_{t \sim [d]} h(\boldsymbol{\mu_t}) -  \Pr_{\boldsymbol{\mu} \sim B} h(\boldsymbol{\mu}) \right| \leq \frac{1}{b} \tau_\Delta.
\end{equation}
Then Algorithm~\ref{algo:weight-esitmator} always returns estimation for $v[i]$ within error $\pm \frac{1}{\sqrt{k}} \| \boldsymbol{v}_{\tail[k]} \|_2$.
\end{lemma}

Before showing the proof, we observe that one can combine Lemma~\ref{obliv-estimate-correct:lemma} with a Chernoff bound to show that our algorithm is accurate in the oblivious setting.

\begin{proof}
First, we consider predicates of the form $h_{i,\rm bad} = \mathbbm{1}_{\mu[i]\ne 0\land \mu[i]\cdot c\notin v[i]\pm  \frac{1}{\sqrt{k}} \|\boldsymbol{v}_{\tail[k]} \|_2}$. Fix an index $i$ and consider $h_{i,\rm bad}$. By Chebyshev's inequality, we have that
\[
\Pr_{\boldsymbol{\mu}\sim \mathcal{B}} \left[ \left| \mu[i]\cdot c - v[i] \right| \in\frac{1}{\sqrt{k}} \|\boldsymbol{v}_{\tail[k]} \| \biggm| \mu[i]\ne 0 \right] \ge \tau_{\rm up}
\]
provided $b\ge C\cdot k$ for some constant $C$. (The proof is similar to that of Lemma~\ref{lemma:heavy}). Recall $\tau_{\rm down} = \frac{1}{10}$ and $\tau_{\rm up} = \frac{9}{10}$. We have
\[
\Pr_{\boldsymbol{\mu}\sim \mathcal{B}}[ h_{i, \rm bad}(\boldsymbol{\mu}) ] \le \frac{\tau_{\rm down} }{b}.
\]

We also consider predicates of the form $h_{i,\rm good} = \mathbbm{1}_{\mu[i]\ne 0 \land \mu[i]\cdot c\in v[i] \pm \frac{1}{\sqrt{k}} \| \boldsymbol{v}_{\tail[k]} \|_2 \land \boldsymbol{\mu} \text{ is useful}}$. by Definition~\ref{sequsefule:def}, we know that
\[
\Pr_{\boldsymbol{\mu}\sim \mathcal{B}}[ \boldsymbol{\mu} \text{ is useful} \mid \mu[i]\ne 0 ] \ge 1 - \tau_{\Delta},
\]
By union bound, we know that
\[
\Pr_{\boldsymbol{\mu}\sim \mathcal{B}}  [ h_{i,\rm good}(\boldsymbol{\mu}) ] \ge (\tau_{\rm up}-\tau_\Delta) \frac{1}{b}.
\]
Now, if the condition in the lemma statement holds, for every $i\in [n]$ we have
\[
\Pr_{t\sim [d]}[h_{i,\rm good}(\boldsymbol{\mu}_t)] \ge \frac{\tau_{\rm up} - 2\tau_\Delta}{b}
\]
and
\[
\Pr_{t\sim [d]}[h_{i,\rm bad}(\boldsymbol{\mu}_t)] \le \frac{\tau_{\rm down} + \tau_{\Delta}}{b}.
\]
We choose $\tau_{\Delta}$ to be less than $\frac{1}{10}$. By doing so, the condition of Lemma~\ref{lemma:good-estimation} is met, and Algorithm~\ref{algo:weight-esitmator} can report estimations within the desired accuracy. This completes the proof.  
\end{proof}

Next, we consider the real execution of Algorithm~\ref{algo:weight-esitmator} in the adaptive setting. Knowing that $Q,\boldsymbol{v}$ is $(\xi, \eta)$-DP with respect to the buckets $(\boldsymbol{\mu}_t)_{t\in [d]}$, we might use the generalization property of DP (Theorem~\ref{theo:DP-generalization}) to argue that the number of good buckets is still close to its expectation in the independent case. So the condition of Lemma~\ref{obliv-estimate-correct:lemma} is still met with high probability. This step is analogous to the argument in Section~\ref{sec:obliv-to-adaptive} and we do not repeat it here.

\section{Query Efficiency} \label{faster:sec}

Our robust estimator as described in Algorithm~\ref{algo:robust-count-sketch} tests all $n$ keys using the $\ThresholdMonitor$ in order to determine the returned set $\texttt{K}$ of candidate heavy hitters. This is computationally inefficient. In the following, we show how we can combine our algorithm with an efficient non-robust heavy-hitters sketch to support heavy-hitter queries that are both efficient and robust.

First, let us recall the following fast oblivious heavy-hitter algorithm from \cite{heavy-hitter-fast-query}:
\begin{theorem}[\cite{heavy-hitter-fast-query}]\label{fast-query:thm}
There is an oblivious streaming algorithm \FastQueryHH~ satisfying the following. For every $n,k\ge 1$ and $\beta \in (0,1)$, $\FastQueryHH$ maintains a vector $\boldsymbol{v}\in \mathbb{R}^n$ using $O(k \log n/\beta)$ space, supports entry updates in $O(\log n/\beta)$ time. On a query, it produces in $O(k\mathrm{polylog}(n))$ time a list $L\subset [n]$ of indices such that:
\begin{itemize}
 \item $|L| \le O(k)$.
 \item For every $i\in [n]$ such that $v[i] \ge \frac{1}{\sqrt{k}} \| \boldsymbol{v}_{\tail[k]} \|_2$, it holds that $i\in K$.
\end{itemize}
\end{theorem}

With $\FastQueryHH$ in hand, we can slightly modify Algorithm~\ref{algo:robust-count-sketch} as follows: Concurrently to an execusion of Algorithm~\ref{algo:robust-count-sketch}, we run a copy of $\FastQueryHH$ with the parameter ``$k$'' in Theorem~\ref{fast-query:thm} set to $(C_a+1)\cdot b$. Then on each heavy-hitter query, we first invoke $\FastQueryHH$ to produce a list $L$ of candidate sets. Then we test each key $i\in L$ by querying the TM (just like what we have done in Algorithm~\ref{algo:robust-count-sketch}). We call the new algorithm $\FastQueryVariant$.

We analyze the query time and correctness of $\FastQueryVariant$. On a heavy-hitter query, $\FastQueryVariant$ first takes $O(b\log n/\beta)$ time to run $\FastQueryHH$ and produces a list $L$. Then it tests each key from $L$ by querying the $\ThresholdMonitor$, which requires $O(b\cdot \ell)$ time for each key. Recall that $b = O(k)$. Therefore, the query time is bounded by $O(\mathrm{poly}(k, \ell, \log n/\beta))$. $\FastQueryVariant$ is also correct with high probability: the final list has size bounded by $O(k)$ and every heavy hitter will appear in the list $L$ and then get selected by the $\ThresholdMonitor$ test with high probability.

Finally, we argue that $\FastQueryVariant$ can robustly answer at least $\Omegatilde(\ell^2)$ adaptively-chosen queries. This is because for every $i\in [n]$ with $|v_i| \le \frac{1}{\sqrt{k}} \| \boldsymbol{v}_{\tail[C_a\cdot b]}\|$, it does not really matter if $i$ is included in the list $L$ or not, for such $i$ will always be rejected by the $\ThresholdMonitor$ test with high probability. Therefore, the \emph{statistical distance} between the output distribution of $\FastQueryVariant$ and that of Algorithm~\ref{algo:robust-count-sketch} is negligible. Hence, $\FastQueryVariant$ basically enjoys the same robustness guarantee as Algorithm~\ref{algo:robust-count-sketch}, and the analysis of Algorithm~\ref{algo:robust-count-sketch} implies $\FastQueryVariant$ is also robust.
 
 \section{Robust Stable Estimator for Continuous Reporting}\label{streaming:sec}
 
 In this section we describe Algorithm~\ref{algo:streaming-robust}, which is a robust version of the stable sign-alignment estimator presented in Section~\ref{sec:basic-stable-sa-estimator}. 
 
 One important class of applications of $\countsketch$ is on distributed or streaming data, which has the form of updates to weights of keys.
 In a streaming model the input is presented as a sequence of updates
 $(\boldsymbol{u}_j)_{j=1}^m$, where each update has the form
 $\boldsymbol{u}=(\boldsymbol{u}.{\mathrm{key}}, \boldsymbol{u}.{\mathrm{value}})$ such that $\boldsymbol{u}.{\mathrm{key}}\in [n]$ and $\boldsymbol{u}.{\mathrm{value}}\in \mathbb{R}$.  
   The stream corresponds to an input vector $\boldsymbol{v}$ that is updated in each step. We can denote that as a sequence $(\boldsymbol{v}_q)$ of input vectors where $\boldsymbol{v}_q$ is the vector at step $q$. Initially, $\boldsymbol{v}_0 = 0^n$ and for step $q\geq 1$
 \[\boldsymbol{v}_q := \boldsymbol{v}_{q-1} + u_q.\mathrm{value} \cdot \boldsymbol{e}_{u_q.\mathrm{key}}\ .\]
  The sketch of streaming data can be efficiently maintained:
 Consider an initialized $\Bcountsketch$ with measurement vectors
 $(\boldsymbol{\mu}_t)$.  The sketch is initialized by setting $(c_t \gets 0)_{t\in [d]}$ and updating
 $(c_t \gets c_t + \mu_t[u.{\mathrm{key}}]\cdot u.{\mathrm{value}})_{t\in [d]}$.

 Algorithm~\ref{algo:streaming-robust} is similar to the robust threshold
 estimator, but it only accounts for {\em changes} in the reporting (in terms of when access counts on the $\TM$ are increased). A change occurs when a key that is not currently reported enters the reported set or when a key that is currently reported exits the reported set. 
  
  \paragraph{Analysis:} The analysis mimics that of Algorithm~\ref{algo:robust-count-sketch}.  The only differences is that we use $\tau_\Delta := (\tau_b-\tau_a)/25$, apply Lemma~\ref{stablesmallerrthencorrect:lemma} for the quality of approximations needed, and define $\lambda_{Q,i}$ as the flip number (Definition~\ref{flipnumber:def}) and use Remark~\ref{flipbound:rem} to establish that it bounds the number of $\TM$ queries returning $\top$ for the respective key. The definition of useful buckets, sequences $\mathcal{Q}(L)$, Lemma~\ref{constuseful:lemma}, and Theorem~\ref{Adaptive:thm} all carry over.
 
 \paragraph{Reporting estimated weights:}
The robust estimator in Algorithm~\ref{algo:streaming-robust} can optionally report estimated weights.
When a key $i$ enters the reported set, we obtain a private estimate $\hat{v}[i]$ of its current value, as described in Section~\ref{estimates:sec}.  While keys are in the reported set, we apply exact updates to their estimated values.  When a key exits the reported set we delete its estimated value. As for the analysis, one $\TM$ query returning $\top$ is used upon an entry of a key to the reported set.  Following that we perform exact updates, which do not affect the privacy and robustness analysis.  The estimates we maintain however need to  be validated as we seek accuracy  with respect to the tail of the {\em current} $\boldsymbol{v}_q$ and the permitted "tail error" might decrease.
This can make the estimate obtained earlier when  key $i$ entered the reported set not sufficiently accurate. Therefore, we
validate our estimates before reporting them and perform an update (a fresh weight estimates) if the validation fails. The validation tests (using predicate queries to $\TM$) if our maintained estimate is still within appropriate noisy quantiles of the bucket estimates.  Each validation failure that triggers an update corresponds to a significant decrease of the tail error.  This number contributes to the robustness budget when weights are reported.

\begin{algorithm2e}
    \caption{Robust $\Bcountsketch$: Streaming}
    \label{algo:streaming-robust}
    \DontPrintSemicolon
    \KwIn{Sketch parameters $(n,d,b)$, Access limit $L$, upper bound $m \geq  L$ on the number of updates.}
    $(\boldsymbol{\mu}_t)_{t\in [d]} \gets$ initialized $\Bcountsketch$ with parameters $(n,d,b)$\;
    $\tau_{m_1} \gets \tau_a + \frac{1}{5} (\tau_b-\tau_a)$\tcp{lower threshold value for $\TM$}
    $\tau_{m_2} \gets \tau_b - \frac{1}{5} (\tau_b-\tau_a)$\tcp{higher threshold value for $\TM$}
    $\texttt{K}^+,\texttt{K}^- \gets \emptyset$\tcp{Initialize continuously reported set $\texttt{K}^+\cup \texttt{K}^-$ of candidate heavy hitter keys}
    
    \SetKwProg{Init}{initialize}{:}{}
    \Init{{\rm $\ThresholdMonitor$} $\TM$}{
    $\TM.S \gets (\boldsymbol{\mu}_t)_{t\in [d]}$ \tcp{Measurement vectors} 
   $\TM.(\eps,\delta) \gets  \left(\frac{C_1}{\sqrt{L}}, \frac{C_2}{n\cdot m \cdot b \cdot L} \right)$\tcp{privacy parameters}
   $\TM.L \gets L$ \tcp{access limit}
    }
     \ForEach(\tcp*[h]{process update})
     {update $\boldsymbol{u}=(u.\mathrm{key},u.\mathrm{value})$ }
     {
     \tcp{optional: If $u.\mathrm{key}\in \texttt{K}^+\cup \texttt{K}^-$ then $\hat{v}[u.\mathrm{key}] += u.\mathrm{value}$}
     \ForEach(\tcp*[h]{update sketch}){$t\in [d]$}{$c_t \gets c_t + \mu_t[u.\mathrm{key}] \cdot u.\mathrm{value}$}
     
    \ForEach(\tcp*[h]{loop over keys not currently reported}){ $i\in [n]\setminus (\texttt{K}^+\cup \texttt{K}^-)$}
    {
    Define $f^+:[d] \to\{0, 1\}$, where $f^+(t) = \mathbbm{1}_{\mu_t[i] \cdot c_t > 0}$\;
    \eIf{$\TM.query(f^+ ,\geq \frac{d}{b} \tau_{m2} ) = \top$}{$\texttt{K}^+ \gets \texttt{K}^+ \cup \{i\}$\tcp{Optional: $\hat{v}[i] \gets$ estimate using Algorithm~\ref{algo:weight-esitmator}}}
    {
    Define $f^-:[d] \to\{0, 1\}$, where $f^-(t)=\mathbbm{1}_{\mu_t[i]\cdot c_t < 0}$\;
    \If{$\TM.query(f^-, \geq \frac{d}{b} \tau_{m2}) = \top$}{$\texttt{K}^-\gets \texttt{K}^-\cup \{i\}$ \tcp{Optional: $\hat{v}[i] \gets$ estimate using Algorithm~\ref{algo:weight-esitmator}}}
    }
        } 
    \ForEach(\tcp*[h]{loop over keys currently reported}){ $\sigma\in \{+,-\}$}
    {
    \ForEach{$i\in \texttt{K}^\sigma$}
    {
    Define $f^\sigma:[d] \to\{0, 1\}$, where $f^\sigma(t)=\mathbbm{1}_{\mu_t[i]\cdot c_t \cdot (\sigma 1) > 0}$\;
    \If{$\TM.query(f^\sigma, \leq \frac{d}{b} \tau_{m1}) = \top$}{$\texttt{K}^\sigma \gets \texttt{K}^\sigma \setminus \{i\}$ \tcp{Optional: delete $\hat{v}[i]$ }}
    }
    }
    \Return{$\{(i,\hat{v}[i]) \mid i\in \texttt{K}^+\cup \texttt{K}^-\}$}\tcp{Optional: test whether $\hat{v}[i]$ is still valid.  If not replace $\hat{v}[i]$ using Algorithm~\ref{algo:weight-esitmator}}
    } 
\end{algorithm2e}

\section{Attack Strategy Overview}\label{sec:attackstrategy}
We describe our attack strategy which applies with
both $\countsketch$ and $\Bcountsketch$.
The attacker knows the parameters of the sketch $(n, d, b)$ but does not know the internal randomness $\boldsymbol{\rho}$ and thus
the measurement vectors $(\boldsymbol{\mu}_t)_{t\in [d]}$.
The attacker has an oracle access to the sketch realization by adaptively producing query vectors $(\boldsymbol{v}_q)$ and obtaining a set of keys:  
\[\texttt{K}(\boldsymbol{v}_q) := M(\sketch_\rho(\boldsymbol{v}_q)) .\]
We tailor our attack strategy to the following estimators $M$:
\begin{trivlist}
\item
{\bf The median estimator:}
Compute for each key $i$ the estimate $\hat{v}[i] \gets \mathrm{median}\{\mu_t[i]\cdot c_t \mid t\in T_i \}$ and return the $k'$ keys with largest estimated magnitudes.
\item
{\bf Basic threshold sign-alignment estimator:}
The estimator in Equation ~\eqref{basicest:eq}, with arbitrary $\tau$, and with the basic estimate as in Equation~\eqref{countagree:eq}.
\item
{\bf Robust threshold sign-alignment estimator:}
As in Algorithm~\ref{algo:robust-count-sketch}.
\end{trivlist}
We establish lower bounds on robustness that nearly match the respective upper bounds:
\begin{theorem}
The median estimator and basic threshold sign-alignment estimator admit an attack of size $O(\ell)$.  The robust estimator admits an attack of size $O(\ell^2)$, where $\ell = d/b$.
\end{theorem}

For $i\in [n]$, we denote by $\boldsymbol{e}_i$ the unit vector of the $i$'th entry (that is, $e_i[j]=\mathbbm{1}_{i=j}$). 
For a vector $\boldsymbol{z}\in\mathbb{R}^n$, the support of $\boldsymbol{z}$, which we denote by $\Supp(\boldsymbol{z}) \subset [n]$, is the set of its nonzero entries. For a key $i\in [n]$, we denote its buckets by $T_i = \{t\in [d] \mid \mu_t[i] \neq 0 \}$.

\paragraph{Attack Description} 
Our attacks will be with respect to one designated key $h\in[n]$ and a specified bias-to-noise ratio parameter $\texttt{BNR} > 0$. We construct an attack tail $\boldsymbol{a}$, which is a vector with entries in $\{-1,0,1\}$ so that $h\not\in\Supp(\boldsymbol{a})$ and support size
$|\Supp(\boldsymbol{a})|=r\cdot m$ (where $r,m$ are attack parameters) that behaves very differently on the buckets $T_h$ than on random buckets $\boldsymbol{\mu}\sim\mathcal{B}$.
On a random bucket we have that
$\E_{\boldsymbol{\mu}\sim \mathcal{B}}\left[ \langle \boldsymbol{\mu},\boldsymbol{a}\rangle \right]= 0$ and 
$\Var_{\boldsymbol{\mu}\sim \mathcal{B}}\left[ \langle \boldsymbol{\mu},\boldsymbol{a}\rangle \right] = \frac{r\cdot m}{b}$ and
standard deviation $\sqrt{\frac{1}{b}\cdot r\cdot m} = \sqrt{\frac{1}{b}}\| \boldsymbol{a} \|_2$.
On most of the sketch buckets $(\boldsymbol{\mu}_t)_{t\in T_h}$, however,  we will have large bias with
 \begin{equation} \label{finalattacktailprop:eq}
 \langle \boldsymbol{\mu}_t, \boldsymbol{a} \rangle \cdot \mu_t[h] \geq   \texttt{BNR} \cdot \frac{1}{\sqrt{b}}\sqrt{r\cdot m}\ . 
 \end{equation} 
We then use $\boldsymbol{a}$ to design a vector on which the sketch fails.  The details depend on the estimator.  When considering
inputs of the form
\[w \cdot \mathbf{e}_h \pm \boldsymbol{a}\ ,\]
we can expect a randomly initialized sketch to provide a good estimate of $w$ when $w\gg \sqrt{\frac{1}{b}\cdot r\cdot m} = \frac{1}{\sqrt{b}} \|\boldsymbol{a}\|_2$ but our attacked sketch masks $w$ and the estimator might fail to
return $h$  even when 
$w \gg \frac{1}{\sqrt{k}}\|\tail_k(\boldsymbol{a})\|_2 \approx \sqrt{\frac{r\cdot m}{k}}$ or
can return it as heavy even when $w=0$.

\paragraph{Construction of $\boldsymbol{a}$}
We construct our attack tail $\boldsymbol{a}$ by identifying and summing up $r=O(\ell)$ randomly-generated tails that are lightly biased in the same direction.   When we sum $r$ lightly-biased tails, the bias grows linearly with $r$  whereas the standard deviation grows proportionally to $\sqrt{r}$.  Therefore the ratio of the bias to the standard deviation increases. 

We identify these lightly biased tails using queries in which $h$ is a borderline heavy hitter and selecting tails on which $h$ is reported (or selecting negated tails where it is not reported). 
We designate a special set $H$ of keys that includes 
$h$ and use
query vectors of the form
\begin{equation} \label{queryvec:eq}
    \sum_{i\in H} v[i] \cdot \mathbf{e}_i \pm \boldsymbol{z}_q\ .
\end{equation}
The values
 $v[i]\in\mathbb{R}_{\geq 0}$ for $i\in H$ are fixed throughout the attack. 

The vectors $\boldsymbol{z}_q \in \{-1,0,1\}^n$ are chosen so that they have 
disjoint supports $\Supp(\boldsymbol{z}_q)$ of size
$|\Supp(\boldsymbol{z})| = m$ and so that
these supports are disjoint from $H$. For each key
$i\in \Supp(\boldsymbol{z}_q)$, the value $z_q[i]$ is drawn independently and uniformly at random from $\{-1,1\}$.   
 Our attack will process the sequence $(\boldsymbol{z}_q)$ and for each, we will either {\em collect} $\boldsymbol{z}_q$ (that is, 
$\boldsymbol{a} \gets \boldsymbol{a} + \boldsymbol{z}_q$)
collect its negation
(that is, 
$\boldsymbol{a} \gets \boldsymbol{a} - \boldsymbol{z}_q$)
or collect neither.  
We stop after we collect a specified number $r$ of tails.
 
 Our attack on the median estimator uses
$|H|=k'+1$ with $k'-1$ very heavy keys and two equal-weight borderline keys, one of which is $h$. 
The estimator outputs all very heavy keys and only one of the two borderline keys. If $h$ is reported we collect
$\boldsymbol{z}$ and if it is not reported we collect $-\boldsymbol{z}$.\footnote{We assume here that it is unlikely that the median estimates of the two borderline keys are equal. This can be assured with $\sqrt{m/b} \gg d/b$ or by using Gaussian values instead of $-1,1$.}  For the sign-alignment 
estimators we use $H= \{h\}$.
We query with both 
$\boldsymbol{z}$ and $-\boldsymbol{z}$.  If $h$ is reported exactly once then either $\boldsymbol{z}$ or its negation are collected accordingly. The querying with the robust estimator, which is probabilistic, requires $O(\ell)$ queries and is discussed below.

\paragraph{Analysis}
The contribution of each $\boldsymbol{z}$ to a bucket in $T_h$ is a symmetric random variable that is the sum of independent random variables that are $0$ with probability $(1-1/b)$ and $1$ or $-1$ each with probability $1/2b$.  The distribution has expectation $0$ and standard deviation $\sqrt{m/b}$. 
We now consider the signs of the tail contributions 
$\mu_t[h] \langle \boldsymbol{\mu}_t,\boldsymbol{z}\rangle$
on all $\ell$ buckets (we use $m$ that is large enough so a contribution of $0$ is unlikely). Then each bucket is positive or negative with probability $0.5$.  Using an anti-concentration bound, $\boldsymbol{z}$ with constant probability would have $\Omega(\sqrt{\ell})$ difference between the number of positive or negative  contributions.
Note that when we condition on the sign of $\mu_t[h] \langle \boldsymbol{\mu}_t,\boldsymbol{z}\rangle$,
the distribution of magnitudes is the same as with an unconditioned bucket.

We now assume we can identify which one of $\pm\boldsymbol{z}$ has more positive contributions and collect it. Then with constant probability its contribution to a given bucket in $T_h$ is positive with probability
$1/2 + \Omega(1/\sqrt{\ell}$ and negative otherwise 
(the distribution of magnitudes conditioned on the sign is the same as without conditioning on the selection).
The process is equivalent as drawing a positive value with probability $\Omega(1/\sqrt{\ell}$ and an unbiased value (distributed as with random $\boldsymbol{z}$) otherwise. 
So with constant probability the collected vector adds in expectation bias of $\sqrt{m/(b\cdot \ell)}$ for each bucket of $T_h$ and variance that is roughly that of an unconditioned bucket $\frac{m}{b}$.  
The expected bias added to a bucket 
after $r$ collection events is then
$\Omega( r \cdot \sqrt{m/b} /\sqrt{\ell})$
(which is also well concentrated when $r\gg \sqrt{\ell}$, since in each round there is $\Omega(1/\sqrt{\ell})$ probability of bias that is $\Omega(\sqrt{m/b})$).
The variance of the unbiased part is bounded from above by that of adding tails $\boldsymbol{z}$ unconditionally, and is  $\leq r \cdot m/b$.

In terms of the "tail norm." Each collection event contributes $m$ entries of magnitude $1$ to $\boldsymbol{a}$. Hence, after $r$ collection events we have 
$\|\boldsymbol{a}\|_2 = \sqrt{r\cdot m}$.  Therefore, the standard deviation
of the per-bucket noise is $\approx \sqrt{r\cdot m/b} = \frac{1}{\sqrt{b}}\|\boldsymbol{a}\|_2$).  

The number $r$ of collection events needed for the bias-to-noise ratio that exceed a specified $\texttt{BNR}$ is (up to constants) the solution of 
\[ 
r \cdot \sqrt{\frac{m}{b\cdot\ell}} = \texttt{BNR}\cdot \sqrt{\frac{r\cdot m}{b}}\ ,
\]
which yields $r = c\cdot \texttt{BNR}^2 \ell$ (for some constant $c$).

What remains is to show how we set up the query vectors with tail $\pm\boldsymbol{z}$ so that we can identify lightly-biased tail.

\paragraph{Choosing the weight of key $h$}
We set the value of $v[h]$ so that collection happens with constant probability for each $\boldsymbol{z}_q$ and that the collected tail is biased in the "right" direction.
To do so, we set  $v[h]$ so that key $h$ 
is "borderline" reported.   For the median and basic threshold estimators, we set it so that over the distribution of $\boldsymbol{z}$ there is constant probability to report $h$ and constant probability not to report $h$.

The robust estimator adds random noise to the count and has a probabilistic output. There is also differences in the count 
due to buckets getting inactivated during the execution. 
So for each count (minus the contribution of inactive buckets) there is a reporting probability.  Recall
that we can expect a gap of $\Omega(\sqrt{\ell})$ between the counts of $\boldsymbol{z}$ and
$-\boldsymbol{z}$.   The reporting probability is $0$ for low counts and $1$ for high counts (close to $\ell$).  Therefore there is a value where the reporting probability increases by at least $\Theta(1/\sqrt{\ell})$ when the count increases by $\Theta(\sqrt{\ell})$.  
Note that both 
the number of inactivated buckets and the DP noise added by $\TM$ are of a lower order that the count gap ($\Theta(1/\ell)$ fraction versus $\Theta(1/\sqrt{\ell})$ gap). Therefore the probability gap holds even when accounting to the DP noise and inactive buckets.

We set $v[h]$ to be in this regime.  
We can then use $O(\ell)$ queries with each of $\boldsymbol{z}$ and
$-\boldsymbol{z}$ to estimate the probabilities to the needed accuracy (additive error of $O(1/\sqrt{\ell})$)  to identify the larger one.

The attacks on the classic and basic estimators use $O(1)$ queries for each collection event so overall the attack uses $O(\ell)$ queries.
The attack on the robust estimator requires $O(\ell)$ queries for each collection event and therefore overall uses
$O(\ell^2)$ queries.

\paragraph{Extensions}
The attacks can be implemented in a streaming model in which case the size of the attack is replaced by the flip number $\lambda_Q$.  We emulate our attacks using updates, noting that $2m$ updates are needed to change between a tail and its negation or between different tails.  This while maintaining stability of the reporting of keys in $H$ (by increasing their values during the updates of the tail).  

The attacks can be generalized to the case when the parameters of the estimators are not known to the attacker.  In this case, we include a "search" for the parameters as part of the attack. 

\section*{Acknowledgements}

Edith Cohen is partially supported by Israel Science Foundation (grant no. 1595/19).
Moshe Shechner is partially supported by Israel Science Foundation (grants no. 1595/19 and 1871/19).
Uri Stemmer is partially supported by the Israel Science Foundation (grant 1871/19) and by Len Blavatnik and the Blavatnik Family foundation.


\bibliographystyle{plain}  
\bibliography{references,robustHH}  

\newpage
\onecolumn
\appendix


\section{Proof of Lemma~\ref{lemma:heavyornot}} \label{proofheavynotheavy:sec}

In this section we provide the proof of Lemma~\ref{lemma:heavyornot}.  We start with some preliminaries.

\subsection{Helpful Inequalities}

\begin{theorem}\label{theo:holder}
(H\"older's inequality)
Let $X, Y$ be two non-negative random variables. Also let $p, q\ge 1$ be two reals such that $\frac{1}{p}+\frac{1}{q}=1$. Then it holds that $\mathbb{E}[X\cdot Y]\le \| X \|_p \cdot \| Y \|_q:= \Ex[X^p]^{1/p}\cdot \Ex[Y^q]^{1/q}$.
\end{theorem}


\begin{lemma}\label{lemma:main-tech}
(technical lemma) Let $A,B\ge 0$ be two reals. Let $X$ be a random variable such that:
\[
\begin{aligned}
\mathbb{E}[X] = 0, \quad & \mathbb{E}[X^2] \ge A, \\
\mathbb{E}[X^3] = 0, \quad  & \mathbb{E}[X^4] \le B.
\end{aligned}
\]
Then, for any $C\ge 0$ one has
\[
\Pr[X \le -C] \ge \frac{1}{\left(16 \cdot (6AC^2+B+C^4)\right)^{1/3}}\left( \left(\frac{(A+C^2)^3}{6AC^2+B+C^4} \right)^{1/2} - C \right)^{4/3}.
\]
\end{lemma}


\begin{proof}
Define a new random variable $Y = X + C$. It suffices to bound $\Pr[Y\le 0]$. First of all, one can verify the following useful moment bounds of $Y$.
\[
\begin{aligned}
& \mathbb{E}[Y] = C, \\
& \mathbb{E}[Y^2] \ge A+C^2, \\
& \Ex[Y^4] \le 6 AC^2 + B + C^4. \\
\end{aligned}
\]
Then, we have
\begin{align}\label{eq:bound-dif}
    \Ex[|Y|\cdot \mathbbm{1}_{Y > 0}] - \Ex[|Y|\cdot \mathbbm{1}_{Y \le 0}] = C.
\end{align}
We also have
\[
\begin{aligned}
\Ex[Y^2]
&= \Ex[ |Y|^{2/3}\cdot |Y|^{4/3}] \\
&\le \left\| |Y|^{2/3} \right\|_{3/2} \cdot \left\| |Y|^{4/3} \right\|_3 & \text{(Theorem~\ref{theo:holder})} \\
&= \Ex[Y]^{2/3} \cdot \Ex[Y^4]^{1/3}, \\
\end{aligned}
\]
which implies that
\begin{align}\label{eq:bound-sum}
 \Ex[|Y|] = \Ex[|Y|\cdot \mathbbm{1}_{Y > 0}] + \Ex[|Y|\cdot \mathbbm{1}_{Y \le 0}] \ge  \left(\frac{\Ex[Y^2]^3}{\Ex[Y^4]}\right)^{1/2}.
\end{align}
Subtracting \eqref{eq:bound-dif} from \eqref{eq:bound-sum}, one gets
\begin{align}\label{eq:lower-bound-Yle0}
    \mathbb{E}[|Y|\cdot \mathbbm{1}_{Y\le 0}] \ge \frac{1}{2}\left( \left(\frac{\Ex[Y^2]^3}{\Ex[Y^4]}\right)^{1/2} - C\right) .
\end{align}
On the other hand, we also have (by Theorem~\ref{theo:holder})
\begin{align}\label{eq:upper-bound-Yle0}
\Ex[|Y|\cdot \mathbbm{1}_{Y\ge 0}] \le \|Y\|_4 \cdot\| \mathbbm{1}_{Y\le 0}\|_{4/3} \le  \Ex[Y^4]^{1/4} \cdot \Pr[Y\le 0]^{3/4}.
\end{align}

Combining \eqref{eq:lower-bound-Yle0} and \eqref{eq:upper-bound-Yle0}, one concludes
\[
\Pr[Y\le 0] \ge \frac{1}{2\cdot (2\cdot \Ex[Y^4])^{1/3}}\left( \left(\frac{\Ex[Y^2]^3}{\Ex[Y^4]}\right)^{1/2} - C\right)^{4/3}
\]
as desired.
\end{proof}

In our analysis, we will instantiate Lemma~\ref{lemma:main-tech} with the following parameter setting.

\begin{corollary}\label{coro:main-tech}
Let $A > 0$ be a real. Let $X$ be a random variable such that:
\[
\begin{aligned}
\mathbb{E}[X] = 0, \quad & \mathbb{E}[X^2] = A, \\
\mathbb{E}[X^3] = 0, \quad  & \mathbb{E}[X^4] \le 3A^2.
\end{aligned}
\]
Then, it holds that $\Pr\left[X \le -\frac{\sqrt{A}}{5}\right]\ge 0.02$.
\end{corollary}

\begin{proof}
By Lemma \ref{lemma:main-tech} and direct calculation. 
\end{proof}

\subsection{Lemma proof details} 
We first recall how a linear measurement of a bucket is sampled. A bucket $\mu\colon [n]\to \{-1,0,1\}$ is constructed as follows: first, sample a $3$-wise independent function $h\colon [n]\to \{0, 1\}$ such that for every $i\in [n]$, $h(i)=1$ with probability $1/b$. Then, sample a $5$-wise independent function $\sigma\colon [n]\to \{-1,1\}$ such that $\sigma(i)$ equals $1$ or $-1$ with equal probability. Finally define $\mu[i]=h(i)\cdot \sigma(i)$. Fix an input vector $\boldsymbol{v}\in \mathbb{R}^n$ and compute $c = \langle \boldsymbol{\mu},\mathbf{v}\rangle$.
The measurement yields the weak estimates $\hat{v}[i] := \mu[i] c$  of $v[i]$ for all $i\in [n]$ that participate in the bucket (that is, $\mu[i] \not= 0$).
\footnote{The Lemma also holds for $\countsketch$.  Note that while buckets of $\countsketch$ can be dependent, the 
$d/b$  buckets that a key $i$ participates in are independent and follow the same distribution as $\Bcountsketch$ buckets.  The only difference is that we have exactly $d/b$ buckets instead of in expectation.}

The first part is established by Lemma~\ref{lemma:heavy} and the second part by Lemma~\ref{lemma:not-heavy}.

\begin{lemma}\label{lemma:heavy}
If $|v[i]| \ge \frac{30}{\sqrt{b}}\|\boldsymbol{v}_{\tail[b/900]}\|_2$, then $\Pr[c\cdot \mu[i] \cdot \mathrm{sign}(v[i]) > 0\mid \mu[i] \ne 0] \ge \frac{399}{400}$.
\end{lemma}

\begin{proof}
We prove for the case that $\sgn(v[i])=1$. The case for $\sgn(v[i]) = -1$ can be handled similarly. Let $S = \{i\in [n]: |v[i]| \text{ is } \frac{b}{900} \text{ largest among all } |v[j]| \text{'s} \}$. Conditioning on $\mu[i]$, we argue as follows. Let $c_{-i} = \sum_{i'\ne i, i'\notin S} \mu[i']\cdot v[i']$. Then one has
\[
\begin{aligned}
\Ex_{\boldsymbol{\mu}}[c_{-i} \mid \mu[i]\ne 0] &=0 , \\
\Ex_{\boldsymbol{\mu}}[c_{-i}^2 \mid \mu[i]\ne 0] &\le \frac{1}{b}\|\boldsymbol{v}_{\tail[b/900]}\|_2^2 .
\end{aligned}
\]
Let $\calE_1$ denote the event that $\mu[i]\cdot c_{-i}<-\frac{30}{\sqrt{b}}\cdot \|\boldsymbol{v}_{\tail[b/900]}\|_2$. By Chebyshev's inequality one has
\[
\Pr_{\boldsymbol{\mu}}[\lnot \calE_1 \mid \mu[i]\ne 0] = \Pr_{\boldsymbol{\mu}}\left[\mu[i] \cdot c_{-i} < -\frac{30}{\sqrt{b}}\cdot \|\boldsymbol{v}_{\tail[b/900]}\|_2 \biggm| \mu[i] \ne 0 \right]\le \frac{1}{900}.
\]
Let $\calE_2$ denote the event that $\mu[i'] = 0$ for all $i'\in S\setminus \{i\}$. We first observe that
\[
\Ex_{\boldsymbol{\mu}}\left[ \sum_{i'\in S\setminus \{i\}} \mathbbm{1}_{\mu[i']\ne 0} \biggm| \mu[i]\ne 0 \right] \le \frac{|S|}{b} \le \frac{1}{900}.
\]
By Markov's inequality, we have
\[
\Pr_{\boldsymbol{\mu}}\left[ \lnot \calE_2\mid \mu[i] \ne 0 \right] = \Pr_{\boldsymbol{\mu}}\left[ \sum_{i'\in S\setminus \{i\}} \mathbbm{1}_{\mu[i']\ne 0} \ge 1 \right] \le \frac{1}{900}.
\]
Note that when both $\calE_1$ and $\calE_2$ hold, we must have $c \cdot \mu[i] = \mu[i]\cdot (c_{-i} + v[i]) > 0$. Therefore, we conclude that
\[
\begin{aligned}
\Pr_{\boldsymbol{\mu}}\left[ c\cdot \mu[i]> 0 \mid \mu[i]\ne 0\right]
&\ge 1 - \Pr_{\boldsymbol{\mu}} \left[\lnot \calE_1 \lor \lnot \calE_2 \mid \mu[i]\ne 0 \right]\ge \frac{399}{400}.
\end{aligned}
\]
\end{proof}

Then, we show if a key is far away from being heavy, then there is a decent chance that the sign of a key does not align with the sign of buckets it participates in. Let $i\in [n]$ be a key. Suppose there are at least $50b$ keys $j$ such that $|v[j]| \ge |v[i]|$. Draw a random bucket $\boldsymbol{\mu}$ conditioning on $\mu[i]\ne 0$. For intuition, let us first consider the case that each coordinate of $\boldsymbol{\mu}$ is \emph{independently} sampled. Then with probability $50\%$ there exists $j\in [n]$ such that $|v[j]| \ge |v[i]|$ and $\mu(j) \ne 0$. Conditioning on this, with probability $\frac{1}{4}$ we have both $\sign(\mu(j)\cdot v[j]) \ne \sign(\mu(i)\cdot v[i])$ and $\sign\left(\sum_{j'\in [n]\setminus \{j,i\}} \mu(j')\cdot v[j'] \right)\ne \sign(\mu(i)\cdot v[i])$. In this case, the sign of the bucket is ``pushed away'' from the direction of $v[i]$. Exploiting moment methods and Lemma~\ref{lemma:main-tech}, we can show a similar conclusion even when $\boldsymbol{\mu}$ is sampled from a limited-independence source. See the lemma below. 

\begin{lemma}\label{lemma:not-heavy}
For every $i\in [n]$, if there are at least $50b$ keys $j$ such that $|v[j]| \ge |v[i]|$, then it holds $\Pr_{\boldsymbol{\mu}}[c\cdot \mu[i] > 0\mid \mu[i] \ne 0] \le 59/60$ and $\Pr_{\boldsymbol{\mu}}[c\cdot \mu[i] < 0\mid \mu[i] \ne 0] \le 59/60$.
\end{lemma}

\begin{proof}
We assume the vector $\boldsymbol{v}$ is sorted in decreasing order of  absolute values. Namely, $|v[1]|\ge |v[2]| \ge \dots \ge |v[n]|$, and we are proving for an item with index $i > 50 k$. This assumption is only for the ease of presentation.

W.L.O.G.\ we condition on $\mu[i]=1$. A similar argument gives the same bound for the case of $\mu[i]= -1$, and one can use the tower rule over two cases and conclude the proof. From now on, we will always condition on $\mu[i] = 1$.

Let $c_{-i}=\sum_{i'\ne i} \mu[i'] \cdot v[i]$. We want to show lower bounds for $\Pr[c_{-i} \le -|v[i]|]$ and $\Pr[c_{-i} \ge |v[i]|]$. In the following we only lower bound $\Pr[c_{-i} \le -|v[i]|]$. The other case is analogous.

Let us define a random variable $X = \sum_{j=1}^{50\cdot b} \mathbbm{1}_{\mu[j] \ne 0}$. That is, we consider the $50\cdot k$ keys of largest magnitude and let $X$ specify how many among them are hashed into the bucket. Then one has
\[
\mathbb{E}[X]  =  50, \quad \Var[X]\le 50.
\]
Let $\calE$ denote the event that $X\ge 25$. Applying a Chebyshev tells us that $\Pr[\calE] \ge 1-\frac{50}{2500} > \frac{9}{10}$.

Note that $X$ and $\calE$ depend only on $h$. Hence, conditioning on $\calE$, $h$ and $\mu[i]=1$, the function $\sigma$ is still $4$-wise independent. Now let us consider $c_{-i}$ (which is now a random variable depending on $\sigma$). In the following, we consider $\Ex_{\sigma}[c_{-i}^p\mid \calE, h]$ for $p\in \{1,2,3,4\}$. First, for $p = 1$ we have
\[
\begin{aligned}
& \Ex_{\sigma}[c_{-i}\mid \calE, h] = \sum_{i'\ne i} v[i] \cdot \Ex_{\sigma} [\mu[i']] = 0. \\
\end{aligned}
\]
For $p = 2$ we have
\[
\begin{aligned}
\Ex_{\sigma}[c_{-i}^2\mid \calE, h] 
&= \sum_{i_1,i_2\in [n]\setminus \{i\} } v[i_1]\cdot v[i_2] \cdot \Ex_{\sigma} [\mu[i_1] \mu[i_2]]  \\
&= \sum_{i_1\in [n]\setminus \{i\} } v[i_1]^2 \cdot \mathbbm{1}_{h(i_1)\ne 0} \\
&\ge 25\cdot v[i]^2.
\end{aligned}
\]
The second line is due to $\sigma(i_1)$ and $\sigma(i_2)$ are independent for $i_1\ne i_2$. The last line is due to our assumption that $X\ge 25$. Then, for $p = 3$ we have
\[
\begin{aligned}
\Ex_{\sigma}[c_{-i}^3\mid \calE, h] 
&= \sum_{i_1,i_2,i_3\in [n]\setminus \{i\} } v[i_1]\cdot v[i_2] \cdot v[i_3] \cdot \Ex_{\sigma} [\mu[i_1] \mu[i_2] \mu[i_3]]  = 0.
\end{aligned}
\]
The second equality holds because for every tuple $(i_1,i_2,i_3)$, there must be an index $i\in [n]$ that appears for \emph{odd} times among $i_1, i_2, i_3$. Suppose $i$ appears $t\in \{1,3\}$ times. We observe that $\Ex_{\sigma}[\mu[i]^t] = 0$. Since $\sigma$ is $3$-wise independent, it follows that $\Ex_{\sigma}[\mu[i_1]\cdot \mu[i_2]\cdot \mu[i_3]] = \Ex[\mu[i]^t] =  0$. Finally, we consider $p=4$, we have:
\[
\begin{aligned}
\Ex_{\sigma}[c_{-i}^4\mid \calE, h] 
&= \sum_{i_1,i_2,i_3,i_4\in [n]\setminus \{i\} } v[i_1]\cdot v[i_2] \cdot v[i_3] \cdot v[i_4] \cdot \Ex_{\sigma} [\mu[i_1] \mu[i_2] \mu[i_3] \mu[i_4]].
\end{aligned}
\]
Following a similar reasoning as we have done above, we observe that a tuple $(i_1, i_2, i_3, i_4)$ may have non-zero contribution, only if there are no indices that appear odd times among $(i_1,i_2,i_3,i_4)$. Then, at least one of the following three conditions must hold: (1) $i_1=i_2,i_3=i_4$, (2) $i_1=i_3,i_2=i_4$ or (3) $i_1=i_4,i_2=i_3$. Therefore, we have:
\begin{align*}
\Ex_{\sigma}[c_{-i}^4\mid \calE, h] 
&\le 3\cdot  \sum_{i_1,i_2\in [n]\setminus \{i\} } v[i_1]^2\cdot v[i_2]^2 \cdot \Ex_{\sigma} [\mu[i_1]^2 \mu[i_2]^2] \\
&\le 3 \sum_{i_1,i_2\in [n]\setminus \{i\} } v[i_1]^2\cdot v[i_2]^2 \mathbbm{1}_{h(i_1)\ne 0} \mathbbm{1}_{h(i_2)\ne 0} \\
&\le 3\left( \sum_{i_1\in [n]\setminus \{i\} } v[i_1]^2 \cdot \mathbbm{1}_{h(i_1)\ne 0} \right)^2 \\
&\le 3\Ex_{\sigma}[c_{-i}^2\mid \calE, h]^2.
\end{align*}

In summary, we have the following.
\[
\begin{aligned}
& \Ex_{\boldsymbol{\mu}}[c_{-i}\mid \calE, h] = \Ex_{\boldsymbol{\mu}}[{c_{-i}}^3] =  0, \\
& \Ex_{\boldsymbol{\mu}}[{c_{-i}}^2 \mid \calE, h] \ge 25 {v[i]}^2,\\
& \Ex_{\boldsymbol{\mu}}[{c_{-i}}^4 \mid \calE, h] \le 3\cdot \Ex[{c_{-i}}^2 \mid \calE, h]^2. \\
\end{aligned}
\]
By Corollary \ref{coro:main-tech}, it follows that
\[
\Pr[c_{-i}\le -v[i]\mid \calE, h]\ge \frac{2}{100}.
\]
Therefore, we have
\[
\Pr_{h,\sigma}[c_{-i}\le -v[i]] \ge \Pr_{h}[\calE]\cdot\Pr[c_{-i}\le -v[i]\mid \calE,h] \ge \frac{2}{100}\cdot \frac{9}{10}\ge \frac{1}{60},
\]
as desired.
\end{proof}

\begin{corollary}\label{coro:not-heavy}
Let $\boldsymbol{v}\in \mathbb{R}^n$ be a vector. For every $i\in [n]$, if $|v[i]| \le \| \boldsymbol{v}_{\tail[50b]}\|$, then $\Pr_{\boldsymbol{\mu}}[c\cdot \mu[i]>0 \mid \mu[i]\ne 0] \le 59/60$ and $\Pr_{\boldsymbol{\mu}}[c\cdot \mu[i] < 0 \mid \mu[i] \ne 0] \le 59/60$.
\end{corollary}

\begin{proof}
$|v[i]| \le \| \boldsymbol{v}_{\tail[50b]} \|$ implies $|v[i]|$ is less than or equal to at least $50b$ other $v[j]$'s. The conclusion follows from Lemma~\ref{lemma:not-heavy} immediately.
\end{proof}

To wrap-up, we combine Lemma~\ref{lemma:heavy} and Corollary~\ref{coro:not-heavy} to justify Lemma~\ref{lemma:heavyornot} with constants $C_a = 50, C_b = 30, \tau_a = \frac{59}{60}$ and $\tau_b = \frac{399}{400}$.

\end{document}